\documentclass[9pt,technotes,web]{ieeecolor}

\usepackage{lcsys}
\usepackage{cite}

\usepackage{url}
\usepackage{mathrsfs,amsmath,amssymb,stmaryrd,amsfonts,mathtools}
\usepackage[linesnumbered,ruled,vlined]{algorithm2e}
\usepackage{graphicx,algpseudocode}
\usepackage{color}
\usepackage{footnote}
\usepackage{generic}
\makesavenoteenv{tabular}
\usepackage{graphicx}
\usepackage{color}

\def\defas{\ensuremath{\mathrel{:=}}}

\def\Set#1#2{\ensuremath{
\left\{#1\,\middle|\,#2\right\}
}}

\def\e{\mathrm{e}}

\def\norm#1{\|  #1 \| }

\def\e#1{\mathrm{e}^{#1} }

\def\abs#1{\left|  #1 \right| }

\def\abs#1{|  #1 | }

\DeclareMathDelimiter{\orbrack}{\mathopen}{operators}{"5D}{largesymbols}{"03}
\DeclareMathDelimiter{\clbrack}{\mathclose}{operators}{"5B}{largesymbols}{"02}
\def\intcc#1{\ensuremath{[#1]}}
\def\intoo#1{\ensuremath{\orbrack#1\clbrack}}
\def\intoc#1{\ensuremath{\orbrack#1]}}
\def\intco#1{\ensuremath{[#1\clbrack}}

\def\Hintcc#1{\ensuremath{\llbracket#1\rrbracket}}

\newcommand{\SO}{\ensuremath{\mathsf{SO(3)}}}

\newcommand{\so}{\ensuremath{\mathfrak{so}(3)}}

\DeclareMathOperator{\tr}{tr}
\DeclareMathOperator{\diag}{diag}

\def\defas{\ensuremath{\mathrel{:=}}}

\newtheorem{theorem}{Theorem}
\newtheorem{lemma}[theorem]{Lemma}

\newtheorem{remark}{Remark}

\newtheorem{corollary}{Corollary}

\newtheorem{proposition}{Proposition}

\title{ \bf  Reach-Avoid Control Synthesis for a Quadrotor UAV\\
with Formal Safety Guarantees}
\author{Mohamed Serry, Haocheng Chang, and Jun Liu \thanks{Mohamed Serry, Haocheng Chang, and Jun Liu are with the Department of Applied Mathematics, University of Waterloo, Waterloo, Ontario, Canada  (e-mail: \{mserry,~h48chang,~j.liu\}@uwaterloo.ca). The first two authors contributed equally to this manuscript. This work was funded in part by the Natural Sciences and Engineering Research Council of Canada and the Canada Research Chairs program.}}

\date{}

\providecommand{\keywords}[1]
{
  \small	
  \textbf{\textit{Keywords---}} #1
}

\begin{document}
\maketitle

\begin{abstract}
Reach-avoid specifications are one of the most common tasks in autonomous aerial vehicle (UAV) applications. Despite the intensive research and development associated with control of aerial vehicles, generating feasible trajectories though complex environments and tracking them with formal safety guarantees remain challenging. In this paper, we propose a control framework for a quadrotor UAV that enables accomplishing reach-avoid tasks with formal safety guarantees. That is, our method yields a desired (position) trajectory to be tracked and an initial set, with respect to the quadrotor's dynamics, such that for any point in that initial set, the resulting quadrotor position follows the desired  trajectory and reaches the target safely while avoiding obstacles, where velocity and thrust safety bounds are satisfied. In this proposed framework, we integrate geometric control theory for tracking and polynomial trajectory generation using B\'ezier curves, where tracking errors are accounted for in the trajectory synthesis process. To estimate the tracking errors, we  revisit the stability analysis of the closed-loop quadrotor system, when geometric control is implemented. We show that the tracking error dynamics exhibit local exponential stability  when  geometric control is implemented with any positive control gains, and we derive tight uniform bounds of the tracking error. We also introduce sufficient conditions to be imposed on the desired trajectory utilizing the derived uniform bounds to  ensure the well-definedness of the closed-loop system.  For the trajectory synthesis, we present an efficient algorithm that enables constructing a safe tube by means of sampling-based planning and safe hyper-rectangular set computations. Then, we compute the trajectory, given as a  piecewise continuous  B\'ezier curve,  through  the safe tube, where a heuristic efficient approach that utilizes iterative linear programming is employed. We  present extensive numerical simulations with a cluttered environment  to illustrate the effectiveness of the proposed framework in reach-avoid planning scenarios.  
\end{abstract}

\section{INTRODUCTION}
Quadrotor UAVs  have been a focal  point in  control theory due to their utilities in a wide variety of applications that involve  surveying and delivering \cite{Raafat2022Application}. One of the main tasks concerning quadrotor applications is reaching specific target(s) while avoiding obstacles, which necessitates adopting and developing safe path planning and control algorithms for quadrotor UAVs.  Standard approaches in the literature rely on constructing  piecewise-continuous polynomial trajectories and then utilizing the differential flatness of  quadrotor UAVs to derive controllers that can track the generated trajectories \cite{mellinger2011minimum}. The generated trajectories are typically constructed by connecting waypoints using  polynomial interpolants. The waypoints, which lay in the three-dimensional operating domain of the quadrotors, can be generated by sampling-based (e.g., RRT \cite{richter2016polynomial}), or node-based (e.g.,  A* \cite{farid2022modified}) methods. Besides reach-avoid problems, waypoints  can also be obtained for more complex specifications such as  temporal logic specifications \cite{pant2018fly}. For polynomial interpolation, standard polynomials \cite{mellinger2011minimum} or  Bernstein polynomials \cite{gao2018online} can be employed and desired trajectories are often obtained by  solving  nonlinear optimization problems.  Various control approaches can then be employed to track the generated trajectories, including linearization-based methods  \cite{hoffmann2008quadrotor,suicmez2014optimal},  feedback-linearization-based  methods \cite{lee2011trajectory,chang2017global}, back-stepping and sliding-mode control  \cite{bouabdallah2005backstepping},    model-predictive control \cite{liu2015explicit, andrien2024model}, and nonlinear feedback control methods on the special Euclidean group \cite{lefeber2017almost,lee2010geometric}. A review of some of the  existing quadrotor control approaches can be found in \cite{amin2016review}. Despite its success, the aforementioned  planning-tracking paradigm lacks formal guarantees in the following sense: if the quadrotor is slightly deviating from the desired planned trajectory initially, it is not guaranteed that the employed controller will ensure tracking the desired path without colliding  with obstacles.

Guaranteeing safety during reach-avoid control scenarios is generally challenging. One approach that provides formal guarantees is abstraction-based control synthesis \cite{reissig2016feedback,zamani2011symbolic,liu2014abstraction,li2022formal}, which is  inapplicable to the high-dimensional quadrotor dynamics as abstraction-based computational requirements increase exponentially with the system dimension.  Another interesting approach that can in theory provides safety guarantees relies on computing control contraction metrics \cite{manchester2017control}, where robust invariant sets are estimated with incorporated feedback control, and outer bounds of the invariant sets are taken into account in the  synthesis of  nominal trajectories to be tracked  \cite{zhao2022tube, singh2023robust, sasfi2023robust}. Such estimates of the invariant sets are computationally demanding, and they rely on polynomial approximations of the quadrotor dynamics; hence, the formal guarantees are compromised. Besides the two aforementioned approaches, formal guarantees may be imposed by incorporating control barrier functions \cite{ames2019control}, which are obtained by solving optimization problems whose feasibility is not guaranteed in general.  Up to our knowledge, incorporating control barrier functions for quadrotor systems has been successfully attained only for simple position and velocity bounds \cite{khan2020barrier}. In summary, formally correct control approaches in the literature are limited when applied to quadrotor systems, which motivates the work presented herein.

To fulfill safety guarantees when synthesizing  a quadrotor reach-avoid controller, we employ the standard planning-tracking paradigm, where we explicitly account for tracking errors when planning a desired trajectory. In principle,  any tracking control method with  theoretical stability guarantees  can be utilized to  derive uniform upper bounds of the tracking errors, and such bounds can then be used in designing  trajectories with formal safety guarantees. Unfortunately, there exist several obstacles that prevent straightforward implementation of the tracking methods in the literature. For example, geometric tracking control \cite{lee2010geometric,lee2015global,fernando2011robust}, a variant of nonlinear feedback control on the special Euclidean group,  was proposed with local exponential stability guarantees that are demonstrated using standard Lyapunov-type analysis. The analyses in \cite{lee2010geometric,lee2015global,fernando2011robust} rely on conservative estimates which are convenient for proving stability but are impractical in planning applications.  Besides,  the theoretical analyses in some of the aforementioned works suffer from technical issues associated with stability guarantees (see Remark \ref{sec:Issue}), and that necessitates revisiting the stability analysis of geometric control. Moreover, the closed-loop quadrotor dynamics, when implementing geometric tracking control, possess a singularity whose analysis was overlooked.  In \cite{lefeber2017almost},  a nonlinear feedback tracking  method was proposed, exhibiting  impressive almost global asymptotic stability guarantees and addressing  the potential singularity in the closed-loop quadrotor system. Unfortunately,   the nature of the  stability analysis in \cite{lefeber2017almost}, which is mainly concerned with demonstrating asymptotic stability, makes it very challenging to derive closed-form error bounds that can be used in planning.  Recently, an extension of the approach in \cite{lefeber2017almost} has been  proposed in \cite{andrien2024model} utilizing model-predictive control to stabilize the translational tracking error, with  similar almost global stability guarantees and, unfortunately, the same issue we highlighted for \cite{lefeber2017almost}. In this work, we aim to adapt the framework of geometric tracking control and revisit its stability analysis to derive closed-form tracking error bounds that can be used in trajectory planning, where we address the issues that we demonstrated above. Interestingly, revisiting the stability of geometric control reveals that the tracking error dynamics, when the geometric controller is implemented with \textit{any} positive control gains, possess local exponential stability properties.

Our contribution in this paper is as follows.
\begin{enumerate}
\item  We present a reach-avoid control synthesis framework for a quadrotor UAV with formal safety guarantees that adapts the standard planning-tracking paradigm, where nonlinear  geometric control is employed for tracking and polynomial trajectory generation is based on B\'ezier curves.   The proposed approach yields a desired (position) trajectory to be tracked and an initial set, where for any point in the initial set, the resulting quadrotor position follows the desired  trajectory and reaches the target safely while avoiding obstacles, where velocity and thrust bounds are respected. 
\item We revisit the stability analysis of the closed-loop quadrotor dynamics when geometric control is integrated.  In our analysis,  we  show that under very mild assumptions (any positive choice for the control gains), the tracking error dynamics of the closed-loop system  exhibit local exponential stability. Our  analysis relies on using nonlinear differential inequalities to prove stability and tight estimates that result in accurate uniform bounds that can be used in the planning procedure. The set of initial values of the quadrotor system that guarantees safety is characterized by a set of nonlinear inequalities. We also introduce sufficient conditions to be imposed on the desired position trajectory for the well-definedness of the closed-loop quadrotor dynamics.
\item We  present an efficient  sampling-based planning algorithm that incorporates safe hyper-rectangular set-based computations to generates a safe tube that can be used in  polynomial trajectory generation. The construction of the safe tube takes into account the uniform bounds of the tracking errors. 
\item 
 We propose a  heuristic efficient approach that computes a safe polynomial trajectory, given as a piecewise-continuous B\'ezier curve,  which  lies within the  generated safe tube, where uniform error bounds are taken into consideration in the synthesis procedure to ensure safety.
\end{enumerate}
The organization of this paper is as follows: the necessary mathematical preliminaries are introduced in Section \ref{sec:Preliminaries}; the quadrotor model equations are presented in Section \ref{sec:Model}; The setup for the reach-avoid control synthesis problem is provided in Section \ref{sec:ProblemFormulation}; the geometric tracking control and the resulting error system are introduced in Section \ref{sec:ErrorDynamics}; stability analysis of the error system is thoroughly investigated in Section \ref{sec:StabilityAnalysis}; the trajectory generation process is presented in detail in Section \ref{sec:TrajectoryGeneration}; the performance of our approach is displayed through numerical simulations in Section \ref{sec:Simulations}; and the study is concluded in Section \ref{sec:Conclusion}. Proofs of some of the technical results presented in this work are included in the Appendix.

\section{Preliminaries}
\label{sec:Preliminaries}
Let $\mathbb{R}$, $\mathbb{R}_+$,  $\mathbb{Z}$, and $\mathbb{Z}_{+}$ denote
the sets of real numbers, non-negative real numbers, integers, and
non-negative integers, respectively, and
$\mathbb{N} = \mathbb{Z}_{+} \setminus \{ 0 \}$.
Let $\intcc{a,b}$, $\intoo{a,b}$,
$\intco{a,b}$, and $\intoc{a,b}$
denote the closed, open and half-open
intervals, respectively, with endpoints $a$ and $b$, and
 $\intcc{a;b}$, $\intoo{a;b}$,
$\intco{a;b}$, and $\intoc{a;b}$ stand for their discrete counterparts,
e.g.,~$\intcc{a;b} = \intcc{a,b} \cap \mathbb{Z}$, and
$\intco{1;4} = \{ 1,2,3 \}$. 

In $\mathbb{R}^{n}
$, the relations $<$, $\leq$, $\geq$, and
$>$ are defined component-wise, e.g., $a < b$, where $a,b\in \mathbb{R}^{n}$, iff $a_i < b_i$ for
all $i\in  \intcc{1;n}$. For $x\in \mathbb{R}^{n}$, $\abs{x}\in \mathbb{R}^{n}$ is defined as $(|x|)_{i}\defas |x_{i}|,~i\in \intcc{1;n}$. The $n$-dimensional vectors with zero and unit entries are denoted by $0_{n}$ and $1_{n}$, respectively. Given two vectors $x,y\in \mathbb{R}^{n}$, $x \cdot y$ denotes the standard inner product of $x$ and $y$ ($x\cdot y=\sum_{i=1}^{n}x_{i}y_{i}$). For 3-dimensional vectors $x$ and $y$,  $x \times y$ denotes the cross product of $x$ and $y$, i.e., 
$$
x\times y\defas\begin{pmatrix}
    x_{2}y_{3}-x_{3}y_{2}\\
    x_{3}y_{1}-x_{1}y_{3}\\
    x_{1}y_{2}-x_{2}y_{1}
\end{pmatrix}.
$$
The space of $n$-dimensional real vectors is equipped with the Euclidean norm $\norm{\cdot}$  ($\norm{x}=\sqrt{x\cdot x},~x\in \mathbb{R}^{n}$). In addition, we make use of the maximal norm $\norm{\cdot}_{\infty}$ defined as $\norm{x}_{\infty}=\max_{i\in \intcc{1;n}}|x_{i}|$, $x\in \mathbb{R}^{n}$.

The $n$-dimensional closed unit balls with respect to $\norm{\cdot}$ and $\norm{\cdot}_{\infty}$ are denoted by $\mathbb{B}_{n}$ and $\mathbb{B}_{n}^{\infty}$, respectively (note that $\mathbb{B}_{n}\subseteq \mathbb{B}_{n}^{\infty}$) . Given $M\subseteq \mathbb{R}^{n}$, $\mathrm{int}(M)$ denotes the interior of $M$, i.e.,  $\mathrm{int}(M)\defas \Set{m\in M}{m+r\mathbb{B}_{n}\subseteq M~\text{for some}~ r>0}$. The convex hull of $N$ points in $\mathbb{R}^{n}$, $x_{1},~x_{2}~,\ldots,~x_{N}$, is denoted by $\mathrm{conv}(x_{1},x_{2},\ldots,x_{N})$. Given $M, N \subseteq \mathbb{R}^{n}$, $M\setminus N$ (set difference of $M$ and $N$) denotes the set  $\Set{x\in M}{x\not\in N}$, $M+N$ (Minkowski sum of $M$ and $N$) denotes the set 
$\Set{ y + z }{ y \in M, z \in N }$, and $M-N$ (Minkowski or Pontryagin difference of $M$ and $N$ (see, e.g., \cite{kolmanovsky1998theory})) denotes the set $\Set{z\in \mathbb{R}^{n}}{z+N\subseteq M}$. Given a finite set $M$, $\mathrm{card}(M)$ denotes the cardinality (number of elements) of $M$. Given $f\colon X\rightarrow Y$ and $C\subseteq X$,~$f(C)\defas \{f(c)|c\in C\}$.

For $a, b \in (\mathbb{R}\cup\{-\infty,\infty\})^ n$ , $a \leq b$,
the  hyper-rectangle $\Hintcc{a,b}$ denotes the set $\Set{x\in \mathbb{R}^{n}}{a\leq x\leq b}$, where, assuming $a$ and $b$ are finite,  $\mathrm{center}(\Hintcc{a,b})\defas(a+b)/2$ and $\mathrm{radius}(\Hintcc{a,b})\defas(b-a)/2$. Note that $\mathbb{B}_{n}^{\infty}=\Hintcc{-1_{n},1_{n}}$.

The identity map (matrix) on $\mathbb{R}^{n}$ is denoted by $\mathrm{I}_{n}$ and the $n\times m$ zero matrix is denoted by $0_{n\times m}$. The transpose of an $n\times m$ matrix $A$ is denoted by $A^{\intercal}$.  For a matrix-valued function $A\colon I\subseteq \mathbb{R}\rightarrow \mathbb{R}^{n\times m}$, we define $A^{\intercal}(t)\defas (A(t))^{\intercal},~t\in I$. The trace and the determinant of an $n\times n$ matrix $A$ are denoted by $\mathrm{tr}(A)$ and $\det(A)$, respectively. Let $d=[d_{1} \cdots {d}_{n}]^{\intercal}\in \mathbb{R}^{n}$, then the diagonal matrix  with diagonal entries $d_{1},\cdots,d_{n}$, is denoted by   $\mathrm{diag}({d})$. 

Let $\mathcal{S}^{n}$ denote the space of $n\times n$ real symmetric matrices, i.e., $\mathcal{S}^{n}\defas \Set{A\in \mathbb{R}^{n\times n}}{A={A}^{\intercal}}$. Note that the eigenvalues of a real symmetric matrix are real. Given $A\in \mathcal{S}^{n}$, $\underline{\lambda}(A)$ and $\overline{\lambda}(A)$ denote the minimum and maximum eigenvalues of $A$, respectively. For $n\times m$ real matrices, $\norm{\cdot}$ is the matrix norm induced by the Euclidean norm ($\norm{A}=\sqrt{\overline{\lambda}(A^{\intercal}A)},~A\in \mathbb{R}^{n\times  m}$).

\subsection{Positive definite matrices}
 Let $\mathcal{S}_{++}^{n}$ denote the set of $n\times n$ real symmetric positive definite matrices $\Set{A\in \mathcal{S}^{n}}{\underline{\lambda}(A)>0}$ (see, e.g., \cite[Chapter~8]{Abadir2005matrix} and \cite[Chapter~6]{horn2012matrix}). Given $A\in \mathcal{S}_{++}^{n}$, $A^{\frac{1}{2}}$ denotes the unique real symmetric positive definite matrix $K$ satisfying $A=K^2$ \cite[p.~220]{Abadir2005matrix}, and we define $A^{-\frac{1}{2}}\defas (A^{\frac{1}{2}})^{-1}$ (this is also equal to $(A^{-1})^{\frac{1}{2}}$  \cite[p.~221]{Abadir2005matrix}). Note that for $A\in \mathcal{S}^{n}_{++}$, $x^{\intercal}Ax=\norm{A^{\frac{1}{2}}x}^{2}$ for all $x\in \mathbb{R}^{n}$. The following lemma states some useful estimates associated with positive definite matrices.
\begin{lemma}[See the proof in the Appendix]\label{Lem:EstimatesPSD}
    Given $M,W\in \mathcal{S}^{n}_{++}$ and $x,\in \mathbb{R}^{n}$, and $A\in \mathbb{R}^{m\times n}$, we have
\begin{enumerate}
\item [(a)] $\underline{\lambda}(M)\norm{x}^{2}\leq x^{\intercal}Mx\leq \overline{\lambda}(M)\norm{x}^{2}$,
 
\item [(b)]$\underline{\lambda}(M^{-\frac{1}{2}}W M^{-\frac{1}{2}})\norm{M^{\frac{1}{2}}x}^{2}\leq \norm{W^{\frac{1}{2}}x}^{2}$,

\item [(c)]$\norm{Ax}\leq \norm{AM^{-\frac{1}{2}}}\norm{M^{\frac{1}{2}}x}$.
\end{enumerate}
\end{lemma}

\subsection{Special orthogonal and skew-symmetric matrices}
Let $\SO$ denote the Lie group of real $3\times 3$ proper orthogonal matrices  and  $\so$ denote the Lie algebra  of real $3\times 3$ skew-symmetric matrices, i.e.,
$
   \SO\defas  \{A\in \mathbb{R}^{3\times 3}|A^{\intercal}=A^{-1},\det (A)=1\}
   $
and 
  $
  \so\defas \{A\in \mathbb{R}^{3\times 3}|A^{\intercal}=-A\}. 
$
The  \textit{hat map}  $\wedge\colon \mathbb{R}^{3}\rightarrow \so$ is defined as  
 $$
 \hat{x} \defas \begin{pmatrix} 0 & -x_3& x_2\\
x_3 & 0 & -x_1\\
-x_2 & x_1 & 0
\end{pmatrix},~x\in \mathbb{R}^{3},
$$ 
and the \textit{vee map} $\vee\colon \so \rightarrow \mathbb{R}^{3}$ is the inverse of $\wedge$. Given a vector-valued function $x\colon I\subseteq \mathbb{R}\rightarrow \mathbb{R}^{3}$, we define $\hat{x}(t)\defas \widehat{x(t)},~t\in I$.   The following lemma demonstrates some properties of the hat and vee maps, which can be easily verified by definition. 
\begin{lemma}\label{lem:PropertiesHat}
 For any $x,y\in \mathbb{R}^{3}$, $A\in \mathbb{R}^{3\times 3}$ and $R\in \SO$, we have
 \begin{align} \label{hat0}
\hat{x}^{\intercal}&=-\hat{x},\\ 
     \hat{x} y &= x \times y = -y \times x = - \hat{y} x,\label{hat1}\\
     \tr[A \hat{x}] &= \frac{1}{2}\tr[\hat{x}(A - A^{\intercal})] = -x^{\intercal}(A - A^{\intercal})^\vee,\label{hat2}\\
     \hat{x}A + A^{\intercal} \hat{x} &= ((\tr[A]I - A)x)^{\wedge},\label{hat3}\\
     R \hat{x} R^{\intercal} &= \widehat{Rx}\label{hat4}.
\end{align} 
 \end{lemma}
 As $\SO$ is compact and connected, the exponential map, $A\mapsto \exp(A)$, from $\so$ to $\SO$ is surjective \cite{hall2013lie}. In addition, the hat map and its inverse establish a homeomorphism between $\mathbb{R}^{3}$ and $\so$. This consequently leads to the Euler-Rodrigues formula  (see, e.g., \cite{gallier2003computing,dai2015euler}):
\begin{lemma}\label{lem:RodriguesFormula}
For each $R\in \SO$, there exists $x\in \mathbb{R}^{3}$ such that 
$$
R=\exp(\hat{x})=\begin{cases} \mathrm{I}_{3}+\frac{\sin(\norm{x})}{\norm{x}}\hat{x}+\frac{1-\cos(\norm{x})}{\norm{x}^2}\hat{x}^{2},&~x\neq 0_{3},
  \\
  \mathrm{I}_{3},&~x=0_{3}.
  \end{cases}
$$
\end{lemma}

\subsection{Differential inequalities}
In our analysis of tracking error stability, we will resort to the following technical results concerning linear and nonlinear differential inequalities. The result below is the well-established Gr\"onwall's lemma (see, e.g., \cite[Lemma~1.1,~p.~2]{bainov2013integral}). 
\begin{lemma}\label{lem:GronWallInequality}
 Let $I\subseteq \mathbb{R}$ be an interval with a left endpoint $t_{0}$, $b\colon I\rightarrow\mathbb{R}$ be continuous and  $u\colon I\rightarrow\mathbb{R}$ be differentiable, satisfying  $
\dot{u}(t)\leq b(t)u(t),~t\in I.
 $
Then, 
 $
u(t)\leq u(t_{0})\e{\int_{t_{0}}^{t}b(s)\mathrm{d}s},~t\in I.
 $
\end{lemma}
The following lemma is  a modified version of \cite[Lemma~4.1,~p.~38]{bainov2013integral}, which is concerned with a Bernoulli-type nonlinear differential inequality that we will utilize in this work. 
\begin{lemma}\label{lem:BernoulliInequality} Let $I\subseteq \mathbb{R}$ be an interval with a left endpoint  $t_{0}$, $b\colon I \rightarrow \mathbb{R}$ be continuous, $k\colon I\rightarrow \mathbb{R}$ be continuous and non-negative, and     $u\colon I\rightarrow \mathbb{R}$ be positive and differentiable, satisfying 
$
\dot{u}(t)\leq b(t)u(t)+k(t)\sqrt{u(t)},~t\in I.
$
Then, for all $t\in I$, 
$$
\sqrt{u(t)}\leq \e{\frac{1}{2}\int_{t_{0}}^{t}b(s)\mathrm{d}s}\left(\sqrt{u(t_{0})}+\frac{1}{2}\int_{t_{0}}^{t}k(s)\e{-\frac{1}{2}\int_{t_{0}}^{s}b(z)\mathrm{d}z}\mathrm{d}s\right).
$$

\end{lemma}
As a consequence of the above lemma, we have the following result.
\begin{lemma}[See the proof in the Appendix]\label{lem:SpecificBernoulliInequality}
    Let $I\subseteq \mathbb{R}_{+}$ be an interval with $0$ as the left endpoint,    $a_{0},a_{1},a_{2},{c}$ be positive constants, and    $u\colon I \rightarrow \mathbb{R} 
    $
    be a positive differentiable function, satisfying
  $$
\dot{u}(t)\leq -(a_{0}-a_{1}\e{-ct})u(t)+a_{2}\e{-ct}\sqrt{u(t)},~t\in I. 
$$
Then,  for all $t\in I$,
\begin{align*} 
\sqrt{u(t)}&\leq     \e{\frac{a_{1}}{2c}}\left(\sqrt{u(0)}\e{-\frac{a_{0}}{2}t}+\frac{a_{2}}{2} \e{-\frac{a_{0}}{2}t}\int_{0}^{t}\e{(\frac{a_{0}}{2}-c)s}\mathrm{d}s\right).
\end{align*}
\end{lemma}

\subsection{B\'ezier curves}
\label{sec:BezierCurveDef}
An $n$-dimensional B\'ezier curve (see, e.g., \cite{farouki2008pythagorean,farin2014curves}) $\mathfrak{B}\colon \intcc{0,T}\rightarrow \mathbb{R}^{n}$, of one segment, with $N+1$ control points, is given by
$$
\mathfrak{B}(t)=\sum_{i=0}^{N}\mathfrak{c}^{i}\mathfrak{b}_{i,N}\left(\frac{t}{T}\right),~t\in \intcc{0,T}, 
$$
where $\mathfrak{c}^{i}\in \mathbb{R}^{n},~i\in \intcc{0;N},$ are the control points determining the shape of the curve, and $\mathfrak{b}_{i,N}\colon \intcc{0,1}\rightarrow \mathbb{R},~i\in \intcc{0;N}$, are   the Bernstein polynomials, which are defined by 
$$
\mathfrak{b}_{i,N}(\bar{t})\defas\frac{i!}{(N-i)!}
\bar{t}^{i}(1-\bar{t})^{N-i},~\bar{t}\in \intcc{0,1},~i\in \intcc{0;N},
$$
where we use the convention that $0^{0}\defas 1$. The initial and final values of $\mathfrak{B}$ coincide with the first and last control points, respectively, i.e.,
$
\mathfrak{B}(0)=\mathfrak{c}^{0},~\mathfrak{B}(T)=\mathfrak{c}^{N}.
$
Using the binomial theorem, it is easily verified that the Bernstein polynomials satisfy  $\sum_{i=0}^{N}\mathfrak{b}_{i,N}\left(\bar{t}\right)=1,~\bar{t}\in \intcc{0,1}$. This induces the useful property
$
\mathfrak{B}(t)\in \mathrm{conv}(\mathfrak{c}^{0},\cdots,\mathfrak{c}^{N}),~t\in \intcc{0,T}.
$
Hence, the range of values of $\mathfrak{B}$ can be controlled by tuning the control points $\mathfrak{c}^{0},\cdots, \mathfrak{c}^{N}.$
In addition, the derivatives of B\'ezier curves are also B\'ezier curves. For example,  
$\dot{\mathfrak{B}}(t)=\sum_{i=0}^{N-1}\frac{N}{T}(\mathfrak{c}^{i+1}-\mathfrak{c}^{i})\mathfrak{b}_{i,N-1}\left(\frac{t}{T}\right)$,
   $\ddot{\mathfrak{B}}(t)=\sum_{i=0}^{N-2}\frac{N(N-1)}{T^2}(\mathfrak{c}^{i+2}-2\mathfrak{c}^{i+1}+\mathfrak{c}^{i})\mathfrak{b}_{i,N-2}\left(\frac{t}{T}\right)$,
     $ \dddot{\mathfrak{B}}(t)=\sum_{i=0}^{N-3}\frac{N(N-1)(N-2)}{T^3}(\mathfrak{c}^{i+3}-3\mathfrak{c}^{i+2}+3\mathfrak{c}^{i+1}-\mathfrak{c}^{i})\mathfrak{b}_{i,N-3}\left(\frac{t}{T}\right)$, and  $ \ddddot{\mathfrak{B}}(t)=\sum_{i=0}^{N-4}\frac{N(N-1)(N-2)(N-3)}{T^4}(\mathfrak{c}^{i+4}-4\mathfrak{c}^{i+3}+6\mathfrak{c}^{i+2}-4\mathfrak{c}^{i+1}+\mathfrak{c}^{i})\mathfrak{b}_{i,N-4}\left(\frac{t}{T}\right)$, $t\in \intcc{0,T}$. Therefore, the range of  values of the derivatives of $\mathfrak{B}$ can also be controlled by tuning the control points $\mathfrak{c}^{0},\cdots, \mathfrak{c}^{N}$.

\section{Quadrotor Model}
\label{sec:Model}
Let $\{{e}_1, {e}_2, {e}_3\}$ be the standard basis of $\mathbb{R}^{3}$ (i.e., $=[e_{1},e_{2},e_{3}]=I_{3}$), which corresponds to the inertial reference (world) frame, where $e_{3}$ is pointing upward. Let   $I\subset\mathbb{R}_{+}$ be an interval  with $0$ as the left endpoint, and let $t\in I$.  The moving fixed-body  frame of the quadrotor system is given by the time-varying orthonormal basis $\{b_{1}(t),b_{2}(t),b_{3}(t)\}$.  The rotation matrix (attitude) that transforms the fixed-body frame into the inertial frame is $R(t)\in \SO$, where $b_{i}(t)=R(t)e_{i},~i\in \intcc{1;3}$. The mass of the quadrotor is $m$ and its inertia matrix, with respect to the moving fixed-body frame, is $J\in \mathbb{R}^{3\times 3}$ (in fact, $J\in \mathcal{S}^{3}_{++}$). The quadrotor system consists of four identical rotors and propellers, with each pair capable of generating thrust and torque independently. The total thrust generated by the propellers is $f(t)\in \mathbb{R}$, and the total generated torque is $\tau(t)\in \mathbb{R}^{3}$. The position of the quadrotor's center of mass is $p(t)\in \mathbb{R}^{3}$, its velocity is $v(t)\in \mathbb{R}^{3}$, and  the quadrotor's angular velocity, with respect to the fixed-body frame, is $\omega(t)\in \mathbb{R}^{3}$. Let $g\in \mathbb{R}_{+}$ denote the gravitational acceleration.  Assuming negligible aerodynamic drag forces\footnote{Control of quadrotors with considerable drag forces has been analyzed in previous works (see, e.g., \cite{lefeber2017almost,andrien2024model}).  The framework presented herein can be adapted to account for drag forces, where the thrust and torque control laws given in \eqref{eqn:f}-\eqref{eqn:M}  are adjusted to account for such forces.}, the governing equations of quadrotor dynamics are 
\begin{align} \label{eq:pdot}
\dot{{p}}(t) &= {v}(t),\\ \label{eq:vdot}
\dot{{v}}(t) &= - g {e}_3 + m^{-1}f(t) R(t) {e}_3,\\ \label{eq:Rdot}
\dot{R}(t)&= R(t) \hat{\omega}(t),\\ \label{eq:wdot}
\dot{{\omega}}(t) &= J^{-1} \left(-{\omega(t)} \times J {\omega(t)} +{\tau(t)}\right),~t\in I.
\end{align}

\section{Problem Formulation}
\label{sec:ProblemFormulation}
 Let $\mathcal{X}_{\mathrm{o}}=\Hintcc{\underline{x}_{\mathrm{o}},\overline{x}_{\mathrm{o}}}\subseteq \mathbb{R}^{3}$ be a bounded hyper-rectangular operating domain and $\mathcal{X}_{\mathrm{u}}=\bigcup_{i=1}^{N_\mathrm{u}}\Hintcc{\underline{x}_{\mathrm{u}}^{(i)},\overline{x}_{\mathrm{u}}^{(i)}}\subseteq \mathbb{R}^{3}$ be  an unsafe set defined as a union of $N_\mathrm{u}$ hyper-rectangles.    It is required that the quadrotor's
  position is always  inside the operating domain, while avoiding the unsafe set. Let $\mathcal{X}_{\mathrm{t}}=\Hintcc{\underline{x}_{\mathrm{t}},\overline{x}_{\mathrm{t}}}\subseteq \mathbb{R}^{3}$ be a hyper-rectangular target set that we aim to drive the quadrotor's position into, where we assume  $\mathcal{X}_{\mathrm{t}}\subseteq \mathcal{X}_{\mathrm{o}}\setminus \mathcal{X}_{\mathrm{u}}$\footnote{The operating domain and the unsafe and target sets should account for the dimensions of the quadrotor to ensure collision avoidance and fully containing the quadrotor within the target region.}.  Given a velocity bound $v_{\max}\in \mathbb{R}^{3}_{+}$, a thrust bound $f_{\max}\in \mathbb{R}_{+}$, a nominal initial point $(p_0,v_{0},R_{0},\omega_{0})\in \mathbb{R}^{3}\times \mathbb{R}^{3}\times \SO\times \mathbb{R}^{3}$, where it is assumed that  $p_{0}\in \mathrm{int}(\mathcal{X}_{\mathrm{o}})\setminus \mathcal{X}_{\mathrm{u}}$, $\abs{v_{0}}<v_{\max}$,
    and  $(R_{0},\omega_{0})=(\mathrm{I}_{3},0_{3})$,
find force and torque control laws such that there exists a finite time $T>0$ and a non-singleton set containing  $(p_{0},v_{0},R_{0},\omega_{0})$ in its relative interior,  where for any initial value of the quadrotor system $(p(0),v(0),R(0),\omega(0))$ in that set, the corresponding trajectory of the quadrotor system and associated thrust satisfy\begin{equation}
\label{eq:ReachAVoidProblem}
\begin{split}
p(t)&\in \mathcal{X}_{\mathrm{o}}\setminus \mathcal{X}_{\mathrm{u}},\\
\abs{v(t)}&\leq v_{\max},\\
\abs{f(t)}&\leq f_{\max},~t\in \intcc{0,T},\\
p(T)&\in \mathcal{X}_{\mathrm{t}}.
\end{split}
\end{equation}
\begin{remark}
   While we have stated previously that geometric tracking control will be employed in the reach-avoid control synthesis, the characterization of the thrust and torque control laws is still incomplete as we need to determine the desired (position) trajectory. In this work, we we will introduce several conditions to be imposed on the desired trajectory to ensure formal safety as formulated in \eqref{eq:ReachAVoidProblem}.   
\end{remark}

\begin{remark}
The framework presented herein can be applied to solve the safe set-stabilization problem
$
p(t) \in \mathcal{X}_{\mathrm{o}}\setminus \mathcal{X}_{\mathrm{u}}
~\forall t\geq 0,$  and  $ \lim_{t\rightarrow \infty} \inf_{x\in \mathcal{X}_{\mathrm{t}}} \norm{p(t)-x}= 0$.  In general, our approach can be employed when considering tracking control on an infinite time horizon, e.g.,  when addressing linear temporal logic \cite{vardi2005automata} specifications.    
\end{remark}

 To address the problem under consideration, our proposed approach relies on designing a desired trajectory with accompanying tracking controller, where the designed trajectory integrates uniform bounds that overestimate the tracking errors. We first introduce the tracking controller and rigorously analyze its stability properties, where we adapt the geometric control framework presented in \cite{lee2010geometric}.

\section{Geometric Control and Error Dynamics}
\label{sec:ErrorDynamics}
Let $I\subset\mathbb{R}_{+}$ be an interval  with $0$ as the left endpoint, and let $t\in I$. For tracking a desired trajectory, with desired position $p_{d}(t)\in \mathbb{R}^{3}$, rotation matrix $R_{d}(t)\in \SO$, and angular velocity $\omega_{d}(t)\in \mathbb{R}^{3}$, the tracking errors between the current and desired states  are defined as follows \cite{lee2010geometric}: 
\begin{align}
{e_p}(t) & \defas {p}(t) - {p}_d(t),\label{eqn:ep}\\
{e_v}(t) & \defas {v}(t) - \dot{p}_d(t),\label{eqn:ev}\\
{e_R}(t) & \defas \frac{1}{2} (R_d^{\intercal}(t)R(t) - R^{\intercal}(t) R_d(t) )^\vee,\label{eqn:eR}\\
{e_\omega}(t) & \defas {\omega} - R^{\intercal}(t) R_d(t) {\omega}_d(t), \label{eqn:eW}
\end{align}
where ${e_p}, {e_v}, {e_R}, {e_\omega}$ are the error functions of the position, velocity, attitude, and angular velocity, respectively. 
We also present the configuration error function \cite{lee2010control,lee2010geometric} 
\begin{equation}\label{eq:Psi}
    \Psi(t) = \frac{1}{2} \mathrm{tr}(\mathrm{I}_{3} - R_d^{\intercal}(t) R(t)),~t\in I.
\end{equation}
We adopt the  geometric control laws for the force and torque
\begin{align}
f(t) & = F_{d}(t) \cdot R(t){e_3},\label{eqn:f} \\
\label{eq:Fd} F_{d}(t)&=-k_{p}e_{p}(t)-k_{v}e_{v}(t)+mg e_{3}+m\ddot{p}_{d}(t),
\\
{\tau}(t) & = - k_R {e_R}(t) -k_{{\omega}} {e_\omega}(t)+{\omega}(t) \times J {\omega}(t)\nonumber\\ 
 &\quad  -J(\hat{\omega}(t) R^{\intercal}(t) R_d(t) {\omega_d}(t) - R^{\intercal}(t) R_d(t){\dot{\omega}_d(t)}),\label{eqn:M}
\end{align}
$t\in I$, where  $k_{{p}}$, $k_{{v}}$, $k_{R}$, and $k_{{\omega}}$ are positive control gains. 
As our problem is concerned with maneuvering within an operating domain to reach a target set while avoiding obstacles, the desired  attitude, and angular velocity should ensure tracking the desired position $p_{d}$ successfully.  As in \cite{lee2010control}, we set 
\begin{align}\label{eq:Rd}
    R_{d}(t)&=[b_{1,d}(t), b_{2,d}(t),b_{3,d}(t)],\\ \label{eq:wd}
\hat{\omega}_{d}(t)&=R_{d}^{\intercal}(t)\dot{R}_{d}(t),
\end{align} 
and  we choose, as in \cite{zou2017trajectory}, 
\begin{equation}
\label{eq:bd1bd2bd3}
\begin{split}
b_{1,d}(t)&=\frac{1}{\norm{F_{d}(t)}}\begin{pmatrix} F_{d,3}(t)+\frac{(F_{d,2}(t))^{2}}{\norm{F_{d}(t)}+F_{d,3}(t)}\\
    -\frac{F_{d,1}(t)F_{d,2}(t)}{\norm{F_{d}(t)}+F_{d,3}(t)}\\
       -F_{d,1}(t) \end{pmatrix},\\
       b_{2,d}(t)&=\frac{1}{\norm{F_{d}(t)}}\begin{pmatrix} -\frac{F_{d,1}(t)F_{d,2}(t)}{\norm{F_{d}(t)}+F_{d,3}(t)},\\
       F_{d,3}(t)+\frac{(F_{d,1}(t))^{2}}{\norm{F_{d}(t)}+F_{d,3}(t)}\\
-F_{d,2}(t)
\end{pmatrix},\\
b_{3,d}(t)&=\frac{F_{d}(t)}{\norm{F_{d}(t)}},~t\in I. 
\end{split}
\end{equation}

We  fix the time interval $I\subseteq \mathbb{R}_{+}$, where we assume  $0$ is the left endpoint, and consider the closed-loop quadrotor system \eqref{eq:pdot}--\eqref{eq:wdot}, with \eqref{eqn:ep}--\eqref{eqn:eW} and \eqref{eqn:f}--\eqref{eq:bd1bd2bd3},  over $I$ throughout the discussion below. Unless otherwise specified, we assume  the control  gains are fixed. In addition, it is  assumed that  $p_{d}$ is four-times continuously differentiable over $I$, where $p_{d}$ and its first four derivatives are assumed to be  uniformly bounded over $I$.

Note that in the definition of $\omega_{d}$ in \eqref{eq:wd}, we differentiate $R_{d}$ with respect to time, which necessarily requires differentiating $F_{d}$ (see equations \eqref{eq:Rd} and \eqref{eq:bd1bd2bd3}), where  the evolution equations of $p$  and $v$, given by \eqref{eq:pdot} and \eqref{eq:vdot}, respectively, are used. Note also that $\omega_{d}$ depends on $\dddot{p}_{d}$ and that $\dot{\omega}_{d}$, which is used in the torque law \eqref{eqn:M}, depends on $\ddddot{p}_{d}$, hence the four-times continuous differentiability assumption imposed on $p_{d}$.

In view of \eqref{eq:bd1bd2bd3}, it is important to ensure that 
\begin{equation}\label{eq:Well-DefinednessRequirement}
F_{d,3}(t)\neq -\norm{F_{d}(t)},~\norm{F_{d}(t)}\neq 0
\end{equation}
 for all  $t\in I$  in order to have the closed-loop dynamics  well-defined. In Section \ref{sec:Well-Definedness}, we will introduce  additional conditions on $p_{d}$ that ensure existence (i.e., the  functions $p$, $v$, $R$, and $\omega$,  defined through the closed-loop system \eqref{eq:pdot}--\eqref{eq:wdot}, with \eqref{eqn:ep}--\eqref{eqn:eW} and \eqref{eqn:f}--\eqref{eq:bd1bd2bd3}, do in fact exists over $I$) and well-definedness (condition \eqref{eq:Well-DefinednessRequirement} holding).   

In the subsequent analysis, we  assume the second derivative of $p_d$ to satisfy  \begin{equation}\label{eq:BoundingAcceleration}
\abs{g{e}_{3}+\ddot{p}_{d}(t)}\leq a_{\max}, 
\end{equation}
for all $~t\in I$ and some specified  $a_{\max}\in \mathbb{R}_{+}^{3}$. This bound will be useful in our stability analysis and when verifying the fulfillment of the thrust bound.

\section{Stability Analysis}
\label{sec:StabilityAnalysis}
In this section, we analyze the evolution of the error terms \eqref{eqn:ep}--\eqref{eqn:eW}. The evolution equations for the error functions are as follows. 
\begin{proposition}[See the proof in the Appendix]\label{Prop:ErrorDynamics}  Assume that the functions $p$, $v$, $R$, and $\omega$, given by \eqref{eq:pdot}--\eqref{eq:wdot}, with \eqref{eqn:ep}--\eqref{eqn:eW} and \eqref{eqn:f}--\eqref{eq:bd1bd2bd3}, exist over $I$ and that  condition \eqref{eq:Well-DefinednessRequirement} holds over $I$. For $t\in I$, the derivatives of the error functions $e_{p}$, $e_{v}$,   $\Psi$, $e_R$, and $e_\omega$ are given by
   \begin{align} \label{eq:ep_dot}
     {\dot{e}_p}(t) & = {e_v}(t),\\ \label{eq:ev_dot}
    {\dot{e}_v}(t) & = -\frac{1}{m}(k_p e_p(t) + k_v e_v(t) - \Delta_{f}(t)),\\
         \label{eq:Psi_dot}  \dot{\Psi}(t) & = e_R(t) \cdot e_\omega(t),\\ \label{eq:eR_dot}
            \dot{e}_R(t) & = \mathcal{C}(t) e_\omega(t),\\ \label{eq:eW_dot}
        {\dot{e}_\omega}(t) & = J^{-1} (-{k_R e_R}(t) - k_{{\omega}} {e_\omega}(t)),
        \end{align}
where  
    \begin{equation} \label{eq:C}
    \mathcal{C}(t) \defas \frac{1}{2}(\tr[R^{\intercal}(t) R_d(t)]\mathrm{I}_{3} - R^{\intercal}(t) R_d(t)),
    \end{equation}
    and  
    \begin{equation}
         \begin{split}
         \Delta_{f}(t)\defas\norm{F_{d}(t)}((b_{3,d}(t)\cdot b_{3}(t))b_{3}(t)-b_{3,d}(t)).
\end{split}
\end{equation}
\end{proposition}

The following technical results will be useful in proving the stability of the error evolution equations. For these results, we assume that the functions $p$, $v$, $R$, and $\omega$, given by \eqref{eq:pdot}--\eqref{eq:wdot}, with \eqref{eqn:ep}--\eqref{eqn:eW} and \eqref{eqn:f}--\eqref{eq:bd1bd2bd3}, exist over $I$ and that  condition \eqref{eq:Well-DefinednessRequirement} holds over $I$. The result below follows by a direct application of the triangle inequality.
\begin{lemma}\label{lem:BoundingForce}
Assume \eqref{eq:BoundingAcceleration} holds and  let $t\in I$. Then, 
$$
\norm{F_{d}(t)}\leq \norm{k_{p}e_{p}(t)+k_{v}e_{v}(t)}+m\norm{a_{\max}}.
$$
    
\end{lemma}

\begin{lemma}[See the proof in the Appendix]\label{lem:BoundingC}
Let $t\in I$. Then, $\norm{\mathcal{C}(t)}\leq 1.
    $
\end{lemma}

\begin{lemma}[See the proof in the Appendix]\label{lem:BoundingPsiUsingEr}
   Let $t\in I$. Then, 
   $\norm{e_{R}(t)}^{2}=\Psi(t)(2-\Psi(t)),$
    and  
   $$
   \frac{1}{2}\norm{e_{R}(t)}^{2}\leq \Psi(t).
   $$
   In addition, if there exists a positive constant $\psi$ such that $\Psi(t)\leq \psi<2$. Then, 
   $$
\Psi(t)\leq  \frac{1}{2-\psi}\norm{e_{R}(t)}^{2}.
   $$
\end{lemma}

\begin{lemma}[See the proof in the Appendix]\label{lem:BoundingTheDifferenceBetweenB3andB3d}
Let $t\in I$. Under the assumption that $\Psi(t)\leq \psi$, for some constant $\psi\in \intoo{0,2}$, we have 
  $$\norm{(b_{3,d}(t)\cdot b_{3}(t))b_{3}(t)-b_{3,d}(t)}\leq \sqrt{\frac{2}{2-\psi}}\norm{e_{R}(t)}.
 $$
\end{lemma}

 Now, we present the first main result of this work concerning the dynamics of the error system \eqref{eq:ep_dot}--\eqref{eq:eW_dot}.

\begin{theorem}\label{thm:LyapunovStability} Assume that the functions $p$, $v$, $R$, and $\omega$, given by \eqref{eq:pdot}--\eqref{eq:wdot}, with \eqref{eqn:ep}--\eqref{eqn:eW} and \eqref{eqn:f}--\eqref{eq:bd1bd2bd3}, exist over $I$ and that  condition \eqref{eq:Well-DefinednessRequirement} holds over $I$. Moreover, assume condition \eqref{eq:BoundingAcceleration} holds over $I$. Let $ \alpha_{\psi}\in \intoo{0,1}$ and $\overline{\Psi}\in \intoo{0,2}$ be specified parameters and assume the following conditions hold:
 \begin{align}   
\label{eq:InitialPsi}
\Psi(0)&\leq \alpha_{\psi}\overline{\Psi},\\
\label{eq:InitialEomega}
\frac{1}{2}e_{\omega}^{\intercal}(0)Je_\omega(0) &\leq   {k_{R}}(1-\alpha_{\psi})\overline{\Psi}.
\end{align}

 Define
\begin{align}
\label{eq:M1}
M_{1}&\defas\frac{1}{2}\begin{pmatrix}
  k_{p}\mathrm{I}_{3} &c_{1} \mathrm{I}_{3}\\
  c_{1}\mathrm{I}_{3} & m \mathrm{I}_{3}
\end{pmatrix},\\
\label{eq:W1}W_{1}&\defas\begin{pmatrix}
\frac{c_{1}k_{p}}{m} \mathrm{I}_{3}&\frac{c_{1}k_{v}}{2m} \mathrm{I}_{3}\\
\frac{c_{1}k_{v}}{2m} \mathrm{I}_{3}& (k_{v}-c_{1})\mathrm{I}_{3}
\end{pmatrix},\\
\label{eq:M2}
M_{2,1}&\defas\frac{1}{2}\begin{pmatrix}
  k_{R}\mathrm{I}_{3} & c_{2}\mathrm{I}_{3}\\
  c_{2} \mathrm{I}_{3}& J
\end{pmatrix},~
M_{2,2}\defas\frac{1}{2}\begin{pmatrix}
  \frac{{2k_{R}}}{2-\overline{\Psi}}\mathrm{I}_{3} & c_{2}\mathrm{I}_{3}\\
  c_{2}\mathrm{I}_{3} & J
\end{pmatrix},\\
\label{eq:W2}
W_{2}&\defas\begin{pmatrix} {c_{2}k_{R}}J^{-1}& \frac{c_{2}k_{\omega}}{2}J^{-1}\\\frac{c_{2}k_{\omega}}{2}J^{-1}&(k_{w}-c_{2})\mathrm{I}_{3} 
\end{pmatrix},
\end{align}
where $c_{1}$ and $c_{2}$ are  constants, satisfying
\begin{align}\label{eq:Boundc1}
 0<c_{1}& < \min\left(\sqrt{k_{p}m}, \frac{4mk_{p}k_{v}}{k_{v}^2+4mk_{p}}\right),\\ \label{eq:Boundc2}
0<c_2&<\min\left(\sqrt{k_{R}\underline{\lambda}(J)}, \frac{4\underline{\lambda}(J)k_{R}k_{\omega}}{k_{\omega}^2+4\underline{\lambda}(J)k_{R}}\right).
\end{align}
Moreover, define
\begin{align*}
 z_{1}(t)\defas
\begin{pmatrix}
    e_{p}(t)\\
    e_{v}(t)
\end{pmatrix},~z_{2}(t)\defas \begin{pmatrix}
    e_{R}(t)\\
    e_{\omega}(t)
\end{pmatrix},~t\in I
    ,
    \end{align*}
 and 
 \begin{align}   
 \label{eq:V1}
V_{1}(t)&\defas z_{1}^{\intercal}(t)M_{1}z_{1}(t),\\
\label{eq:V2}
V_{2}(t)&\defas\frac{1}{2}e_{\omega}^{\intercal}(t) J e_{\omega}(t)+k_{R}\Psi(t)+c_{2}e_{R}(t)\cdot e_{\omega}(t),\\
\label{eq:V}
V(t)&\defas V_{1}(t)+V_{2}(t),~t\in I.
 \end{align}
 Then, 
$M_{1},W_{1},M_{2,1},M_{2,2},W_{2}\in \mathcal{S}^{6}_{++}$.  Furthermore, for all $t\in I$,   
\begin{align}
\label{eq:V2PositiveDefinite}
z_{2}^{\intercal}(t)  M_{2,1}z_{2}(t) \leq V_{2}(t)&\leq z_{2}^{\intercal}(t)  M_{2,2}z_{2}(t),
\end{align}
\begin{align}\label{eq:V2Bound}
  V_{2}(t)&\leq V_{2}(0)\e{-2\beta t},\\
\label{eq:VBound}
    \sqrt{V(t)}&\leq \mathcal{L}(V_{1}(0),V_{2}(0),t ),
   \end{align}
   where, for $x,y,t\in \mathbb{R}_{+}$, 
   \begin{align} \label{eq:L}
    \mathcal{L}(x,y,t )&\defas \mathcal{L}_{1}(x,y,t)
+\mathcal{L}_{2}(y,t),\\
\label{eq:L1}
   \mathcal{L}_{1}(x,y,t)&\defas\e{\frac{\alpha_{1}\sqrt{y}}{2\beta}}\sqrt{x+y}\e{-\frac{\alpha_{0}}{2}t},
   \\
   \label{eq:L2}
   \mathcal{L}_{2}(y,t)&\defas \e{\frac{\alpha_{1}\sqrt{y}}{2\beta}}\frac{\alpha_{2}\sqrt{y}}{2} \e{-\frac{\alpha_{0}}{2}t}\int_{0}^{t}\e{(\frac{\alpha_{0}}{2}-\beta)s}\mathrm{d}s,
   \end{align}
and $\beta$, $\alpha_{0}$, $\alpha_{1}$, and $\alpha_{2}$, are given by equations \eqref{eq:beta}, \eqref{eq:alpha0}, \eqref{eq:alpha1}, and \eqref{eq:alpha2}, respectively.
\begin{figure*}
\begin{align}\label{eq:beta}
   \beta&\defas \frac{\underline{\lambda}(M_{2,2}^{-\frac{1}{2}}W_{2} M_{2,2}^{-\frac{1}{2}})}{2},\\  \label{eq:alpha0}
\alpha_{0}&\defas\min\left(\underline{\lambda}(M_{1}^{-\frac{1}{2}}W_{1} M_{1}^{-\frac{1}{2}}),2\beta\right),\\
 \label{eq:alpha1} \alpha_{1}&\defas\norm{[\frac{c_{1}}{m}\mathrm{I}_{3},\mathrm{I}_{3}]M_{1}^{-\frac{1}{2}}}\norm{[k_{p}\mathrm{I}_{3},k_{v}\mathrm{I}_{3}]M_{1}^{-\frac{1}{2}}}\norm{[\mathrm{I}_{3}, 0_{3\times 3}]M_{2,1}^{-\frac{1}{2}}}\sqrt{\frac{2}{2-\overline{\Psi}}},\\ \label{eq:alpha2}
   \alpha_{2}&\defas m\norm{a_{\max}}\norm{[\frac{c_{1}}{m}\mathrm{I}_{3},\mathrm{I}_{3}]M_{1}^{-\frac{1}{2}}}\norm{[\mathrm{I}_{3},0_{3\times 3}]M_{2,1}^{-\frac{1}{2}}}\sqrt{\frac{2}{2-\overline{\Psi}}}.
\end{align}
\end{figure*}
\end{theorem}

\begin{proof}
    The proof presented herein is based on the proof presented in the archived version of  \cite{lee2010geometric} with significant refinements and adaptations. 
The conditions on $c_{1}$ and $c_{2}$, given by \eqref{eq:Boundc1} and \eqref{eq:Boundc2}, respectively, ensure  the positive definiteness of  $M_{1}$, $W_{1}$, $M_{2,1}$,  $M_{2,2}$, and  $W_{2}$ and that can verified  through Schur complement \cite[Theorem~1.12,~p.~34]{zhang2006schur}. 

Next, we analyze  $V_{2}$.    Define $\tilde{V}_2\colon I\rightarrow \mathbb{R}_{+}$ as
$$
\tilde{V}_2(t)= \frac{1}{2} e_\omega(t) \cdot J e_\omega(t) + k_{R}\Psi(t),~t\in I.
$$
Then, considering  \eqref{eq:Psi_dot} and \eqref{eq:eW_dot}, we have 
\begin{align*} 
\dot{\tilde{V}}_{2}(t)
& = e_\omega(t) \cdot J \dot{e}_\omega(t) + k_{R}\dot{\Psi}(t)\\
         & = e_\omega(t) \cdot (-k_R e_R(t) - k_\omega e_\omega(t)) + k_{R}e_R(t) \cdot e_\omega(t)\\
          &=-k_{\omega} e_\omega(t) \cdot  e_\omega(t)\leq0,~t\in I.
          \end{align*}
This shows that $\tilde{V}_{2}$ is a decreasing function of time, where  $\tilde{V}_{2}(t)\leq \tilde{V}_{2}(0),~t\in I$.  Then, by considering  conditions \eqref{eq:InitialPsi} and \eqref{eq:InitialEomega}, we get  
\begin{align*}
\tilde{V}_{2}(0)=& \frac{1}{2}e_{\omega}(0)\cdot J e_{\omega}(0)+k_{R}\Psi(0)
\\ 
& \leq k_{R}(1-\alpha_{\psi})\overline{\Psi}+k_{R}\alpha_{\psi}\overline{\Psi}=k_{R}\overline{\Psi}.
\end{align*}
Therefore, and as $\Psi(\cdot)$ and $e_{\omega}^{\intercal}(\cdot)J e_{\omega}(\cdot)$ are nonnegative functions of time, 
$\Psi(t)\leq  {\tilde{V}_{2}(t)}/{k_{R} } \leq {\tilde{V}_{2}(0)}/{k_{R} }\leq \overline{\Psi}< 2 $
for all $t\in I$. So, with the assumptions on $e_{\omega}(0)$ and $\Psi(0)$ given in \eqref{eq:InitialEomega} and \eqref{eq:InitialPsi}, respectively, we have, using Lemma \ref{lem:BoundingPsiUsingEr}, 
$$
\frac{1}{2}\norm{e_{R}(t)}^2\leq \Psi(t)\leq \frac{1}{2-\overline{\Psi}}\norm{e_{R}(t)}^{2},~t\in I,
$$ 
and consequently, equation \eqref{eq:V2PositiveDefinite} follows.

Now, we illustrate the exponential decay of $V_{2}$. Differentiating $V_{2}$ with respect to time, where we substitute equations \eqref{eq:eR_dot}, \eqref{eq:eW_dot},  and \eqref{eq:Psi_dot} in, yields
\begin{align*}
\dot{V}_{2}(t)=&e_{\omega}(t)\cdot(-k_{R}e_{R}(t)-k_{\omega}e_{\omega}(t))+k_{R}e_{R}(t)\cdot e_{\omega}(t)\\
&+c_{2}(\mathcal{C}(t)e_{\omega}(t))\cdot e_{\omega}(t)\\
&+c_{2}e_{R}(t)\cdot J^{-1}(-k_{R}e_{R}(t)-k_{\omega}e_{\omega}(t)),~t\in I.
\end{align*}
Using Lemma \ref{lem:BoundingC} and recalling the definition of $W_{2}$ in \eqref{eq:W2}, we have   
\begin{align*}
\dot{V}_{2}(t)\leq&  -e_{R}(t)\cdot(c_{2}k_{R}J^{-1})e_{R}(t)-(k_{\omega}-c_{2})e_{\omega}(t)\cdot e_{\omega}(t)\\
&-e_{R}(t)\cdot(c_{2}k_{\omega}J^{-1})e_{\omega}(t)=-z_{2}^{\intercal}(t)W_{2}z_{2}(t),~t\in I.
\end{align*}
The above estimate can be rewritten as 
$
\dot{V}_{2}(t)\leq - \norm{W_{2}^{\frac{1}{2}}z_{2}(t)}^{2},~t\in I,$ and  relation \eqref{eq:V2PositiveDefinite} can be written as  $ \norm{M_{2,1}^{\frac{1}{2}}z_{2}(t)}^{2}\leq V_{2}(t)\leq \norm{M_{2,2}^{\frac{1}{2}}z_{2}(t)}^{2},~t\in I.
$
Therefore, using Lemma \ref{Lem:EstimatesPSD}(b),  
$
\dot{V}_{2}(t)\leq  - \underline{\lambda}(M_{2,2}^{-\frac{1}{2}}W_{2} M_{2,2}^{-\frac{1}{2}})V_{2}(t),~t\in I.
$
Using Lemma \ref{lem:GronWallInequality} and recalling the definition of $\beta$ in equation \eqref{eq:beta}, 
we  obtain \eqref{eq:V2Bound}. As a consequence of the bounds  \eqref{eq:V2PositiveDefinite} and \eqref{eq:V2Bound}, and using Lemma \ref{Lem:EstimatesPSD}(c), we have
\begin{equation}\nonumber
\begin{split}
\norm{e_{R}(t)} &\leq \norm{[\mathrm{I}_{3},  0_{3,3}]M_{2,1}^{-\frac{1}{2}}}\sqrt{V_{2}(t)}\\
&\leq \norm{[\mathrm{I}_{3},  0_{3,3}]M_{2,1}^{-\frac{1}{2}}}\sqrt{V_{2}(0)}\e{-\beta t},~t\in I.
\end{split}
\end{equation} 

Now, we analyze $V_{1}$.  Note that $V_{1}$ is nonnegative as a consequence of the positive definiteness of $M_{1}$.  Next, we obtain a bound on the derivative of $V_{1}$. Differentiating $V_{1}$ with respect to time, considering \eqref{eq:ep_dot} and \eqref{eq:ev_dot}, yields
\begin{align*}
 \dot{V}_{1}(t)=&-\frac{c_{1}k_{p}}{m}\norm{e_{p}(t)}^{2}-(k_{v}-c_{1})\norm{e_{v}(t)}^{2}\\
 &-\frac{c_{1}k_{v}}{m}e_{p}\cdot e_{v}(t)+\Delta_{f}(t)\cdot\left(\frac{c_{1}}{m}e_{p}(t)+e_{v}(t)\right),~t\in I.
\end{align*}
By recalling the definition of $W_{1}$ in \eqref{eq:W1}, using the triangular inequality, and bounding $\norm{\Delta_{f}(\cdot)}$ utilizing  Lemmas \ref{lem:BoundingForce} and \ref{lem:BoundingTheDifferenceBetweenB3andB3d}, we have
 \begin{align*}
\dot{V}_{1}(t) \leq& -z_{1}^\intercal(t) W_{1}z_{1}(t)+\norm{\frac{c_{1}}{m}e_{p}(t)
+e_{v}(t)}\times\\&\sqrt{\frac{2}{2-\overline{\Psi}}}(\norm{k_{p}e_{p}(t)+k_{v}e_{v}(t)}+m\norm{a_{\max}})\norm{e_{R}(t)},
\end{align*}
$~t\in I$.  Using Lemma \ref{Lem:EstimatesPSD}(c) and the definition of $V_{1}$, we have 
\begin{align*}
 \norm{\frac{c_{1}}{m}e_{p}(t)+e_{v}(t)} &\leq\norm{[\frac{c_{1}}{m}\mathrm{I}_{3},\mathrm{I}_{3}]M_{1}^{-\frac{1}{2}}} \sqrt{V_{1}(t)},\\ \norm{k_{p}e_{p}(t)+k_{v}e_{v}(t)}&\leq \norm{[k_{p}\mathrm{I}_{3},k_{v}\mathrm{I}_{3}]M_{1}^{-\frac{1}{2}}} \sqrt{V_{1}(t)},~t\in I. 
\end{align*}
Moreover, using Lemma \ref{Lem:EstimatesPSD}(b) and the definition of $V_{1}$, we have
$$
-z_{1}^{\intercal}(t)W_{1}z_{1}(t)\leq -\underline{\lambda}(M_{1}^{-\frac{1}{2}}W_{1} M_{1}^{-\frac{1}{2}})V_{1}(t),~t\in I.
$$
Subsequently, and using the definitions of $\alpha_{1}$ and $\alpha_{2}$ in equations \eqref{eq:alpha1} and \eqref{eq:alpha2},  respectively,  and the bound on $\norm{e_{R}(\cdot)}$ using $\sqrt{V_{2}(\cdot)}$, the derivative of $V_{1}$ can be estimated as 
\begin{align*}
 \dot{V}_{1}(t)\leq&-\underline{\lambda}(M_{1}^{-\frac{1}{2}}W_{1} M_{1}^{-\frac{1}{2}})V_{1}(t)\\
&+\alpha_{1}V_{1}(t) \sqrt{V_{2}(t)}+\alpha_{2}\sqrt{V_{1}(t)}\sqrt{V_{2}(t)},~t\in I.
\end{align*}

Now we consider the dynamics of $V$. Note that the non-negativity of $V_{1}$ and $V_{2}$  implies the non-negativity of $V$. In addition, we note that, using the definition of $\alpha_{0}$ in \eqref{eq:alpha0}, 
\begin{align*}
-\underline{\lambda}(M_{1}^{-\frac{1}{2}}W_{1} M_{1}^{-\frac{1}{2}})V_{1}(t)-\underline{\lambda}(M_{2,2}^{-\frac{1}{2}}W_{2} M_{2,2}^{-\frac{1}{2}})V_{2}(t)& \leq\\
-\alpha_{0}(V_{1}(t)+V_{2}(t))= 
 -\alpha_{0}V(t),~t\in I.
\end{align*}
Combining the estimates of $\dot{V}_{1}$ and $\dot{V}_{2}$ and using the fact that $V_{1}(\cdot)\leq V(\cdot)$, the derivative of $V$ can be estimated as
\begin{align*}  
\dot{V}(t)
\leq& -\underline{\lambda}(M_{1}^{-\frac{1}{2}}W_{1} M_{1}^{-\frac{1}{2}})V_{1}(t)-\underline{\lambda}(M_{2,2}^{-\frac{1}{2}}W_{2} M_{2,2}^{-\frac{1}{2}})V_{2}(t)\\
&+\alpha_{1}V_{1}(t) \sqrt{V_{2}(t)}+\alpha_{2}\sqrt{V_{1}(t)}\sqrt{V_{2}(t)}
\\
\leq&  -\alpha_{0}V(t)
+\alpha_{1}V(t) \sqrt{V_{2}(t)}+\alpha_{2}\sqrt{V(t)}\sqrt{V_{2}(t)},~t\in I.
\end{align*}
Assume that for some $t_{0}\in I$, $V(t_{0})=0$. This implies that $V_{2}(t_{0})=0$, and as $V_{2}(t)$ is non-negative satisfying the condition of Gr\"onwall's lemma, Lemma \ref{lem:GronWallInequality} tells us that $V_{2}(t)=0$ for all $t\in I\cap \intco{t_{0},\infty}$. This consequently indicates that $\dot{V}({t})\leq -\alpha_{0} V(t),~t\in I\cap \intco{t_{0},\infty}$, and using the non-negativity of $V$ and Lemma \ref{lem:GronWallInequality}, we have $V(t)=0,~t\in I\cap \intco{t_{0},\infty}$. Therefore, we may assume without loss of generality that $V(t)>0$ for all $t\in I$. Using the bound on ${V_{2}(\cdot)}$ in \eqref{eq:V2Bound} and rearranging we get 
$$
\dot{V}(t)\leq -(\alpha_{0}-\alpha_{1}\sqrt{V_{2}(0)}\e{-\beta t})V(t)+\alpha_{2}\sqrt{V_{2}(0)}\e{-\beta t}\sqrt{V(t)},
$$
$t\in I$. Using Lemma \ref{lem:SpecificBernoulliInequality}, the proof is complete.

\end{proof}

\subsection{Local exponential stability}
Herein, we state an important implication  of Theorem \ref{thm:LyapunovStability}. Assume $I=\mathbb{R}_{+}$. Note that the bounding function $\mathcal{L}$ given in \eqref{eq:L} is continuous over $\mathbb{R}_{+}^{3}$ and  monotonically increasing with respect to the first two arguments, where  $\mathcal{L}(0,0,t)=0$ for all  $t\in I$.  Moreover, for fixed $x,y\in \mathbb{R}_{+}$, $\mathcal{L}(x,y,t)\rightarrow 0$ as $t\rightarrow \infty$. By bounding the function $\mathcal{L}$, we can show that the zero of the error system is locally exponentially stable. Let $x,y,t\in \mathbb{R}_{+}$.  Knowing that $\alpha_{0}/2$, defined by \eqref{eq:alpha0}, is less than or equal to  $\beta$ given in \eqref{eq:beta}, the function $\mathcal{L}_{2}$ in \eqref{eq:L2} can be bounded as follows:
\begin{align*}
  \mathcal{L}_{2}(y,t) = & \e{\frac{\alpha_{1}\sqrt{y}}{2\beta}}\frac{\alpha_{2}\sqrt{y}}{2} \e{-\frac{\alpha_{0}}{2}t}\int_{0}^{t}\e{(\frac{\alpha_{0}}{2}-\beta)s}\mathrm{d}s\\
  & \leq \e{\frac{\alpha_{1}\sqrt{y}}{2\beta}}\frac{\alpha_{2}\sqrt{y}}{2} \e{-\frac{\alpha_{0}}{2}t}\int_{0}^{t}\e{\nu s}\mathrm{d}s\\
  & \leq \e{\frac{\alpha_{1}\sqrt{y}}{2\beta}}\frac{\alpha_{2}\sqrt{y}}{2\nu } \e{-(\frac{\alpha_{0}}{2}-\nu)t},
\end{align*}
where $\nu \in \intoo{0,\alpha_{0}/2}$. 
We can also bound $\mathcal{L}_{1}$ in \eqref{eq:L1} as follows:
$$
\mathcal{L}_{1}(x,y,t)\leq \e{\frac{\alpha_{1}\sqrt{y}}{2\beta}}\sqrt{x+y}\e{-(\frac{\alpha_{0}}{2}-\nu)t}.
$$
Consequently, we have 
$$
\mathcal{L}(x,y,t)\leq \e{\frac{\alpha_{1}\sqrt{y}}{2\beta}}\left(\sqrt{x+y}+\frac{\alpha_{2}\sqrt{y}}{2\nu}\right)\e{-(\frac{\alpha_{0}}{2}-\nu)t},
$$
and, using Theorem \ref{thm:LyapunovStability},
$$
\sqrt{V(t)}\leq \sqrt{V(0)}\e{\frac{\alpha_{1}\sqrt{V(0)}}{2\beta}}\left(1+\frac{\alpha_{2}}{2\nu}\right)\e{-(\frac{\alpha_{0}}{2}-\nu)t},~t\in I.
$$
The above inequality is valid whenever $(e_{p}(0),e_{v}(0), e_{R}(0),e_{\omega}(0))$  satisfies   \eqref{eq:InitialPsi},  and \eqref{eq:InitialEomega}, which implies, with the help of the definitions of $V_{1}$, $V_{2}$, and $V$ given by \eqref{eq:V1}--\eqref{eq:V}, the positive definiteness of $M_{1}$ and $M_{2,1}$, property \eqref{eq:V2PositiveDefinite}, and Lemma \ref{Lem:EstimatesPSD}(a),  the local exponential stability of the error system \eqref{eq:ep_dot}- \eqref{eq:eW_dot}; see, e.g., \cite[Chapter~5]{sastry2013nonlinear}. 
While a similar result was derived in previous works, e.g., \cite{lee2010geometric},  additional assumptions were imposed on the control gains and the parameters $c_{1}$ and $c_{2}$. Our result herein is stronger in the sense that we have local exponential stability for any positive choice of the control  gains $k_{p}$, $k_{v}$, $k_{R}$, and $k_{\omega}$. 

\begin{remark}
\label{sec:Issue}
In an archived and detailed version of \cite{lee2010geometric}, proving the  exponential decay of the function $V$ was attempted by  asserting that, under some technical assumptions, $\norm{e_{v}(\cdot)}$ is time-bounded. That assertion is based on differentiating the function $\tilde{V}_{1}(\cdot)=\frac{1}{2}me_{v}(\cdot)\cdot e_{v}(\cdot)$, a version of $V_{1}$ with $c_{1}$ and $k_{p}$ are set to be zero, and showing that $\dot{\tilde{V}}_{1}$ is negative when $\norm{e_{v}(\cdot)}$ is sufficiently large. Unfortunately, the estimate used in demonstrating the negativity of $\dot{\tilde{V}}_{1}$ is incorrect as it was derived without considering the fact that $\dot{e}_{v}$ should satisfy \eqref{eq:ev_dot}, where $k_{p}$ is still present, and that compromises the correctness of the stability proof.
  \end{remark}
  
\subsection{Uniform bounds}
In this section, we derive, using Theorem \ref{thm:LyapunovStability}, some uniform bounds that  can then be used in the  synthesis of $p_{d}$. Note that  $\mathcal{L}$ given in \eqref{eq:L}, when the first two arguments are fixed, attains a maximum at a finite $t$ that can be computed analytically. We consequently  have the following uniform bound.
\begin{corollary} \label{cor:V_bound} Assume that the functions $p$, $v$, $R$, and $\omega$, given by \eqref{eq:pdot}--\eqref{eq:wdot}, with \eqref{eqn:ep}--\eqref{eqn:eW} and \eqref{eqn:f}--\eqref{eq:bd1bd2bd3}, exist  over $I$ and that  condition \eqref{eq:Well-DefinednessRequirement} holds over $I$. In addition, assume \eqref{eq:BoundingAcceleration},  \eqref{eq:InitialPsi}, and \eqref{eq:InitialEomega} hold. Let the functions $V_{1}$, $V_{2}$, and $V$ be defined as  in \eqref{eq:V1}--\eqref{eq:V}, where conditions \eqref{eq:Boundc1} and \eqref{eq:Boundc2} hold.  Define, for $x,y\in \mathbb{R}_{+}$,
\begin{equation}\label{eq:L_u}
\mathcal{L}_{u}(x,y)\defas \max_{t\in \mathbb{R}_{+}}\mathcal{L}(x,y,t)= \mathcal{L}(x,y,t_{m}(x,y)),
\end{equation}
where $t_{m}$ is given by \eqref{eq:Tm}.  Then,
 $$
\sqrt{V(t)}\leq  \mathcal{L}_{u}(V_{1}(0),V_{2}(0)), t\in I.
 $$

\begin{figure*}
\begin{equation}\label{eq:Tm}
{t}_{m}(x,y)\defas \begin{cases} \max\left(\frac{1}{\frac{\alpha_{0}}{2}-\beta}\ln\left(\frac{\frac{\alpha_{2}\beta \sqrt{y}}{2\beta-\alpha_{0}}}{\frac{\alpha_{0}}{2}\sqrt{x+y}+\frac{\alpha_{0}\alpha_{2}\sqrt{y}}{2(2\beta-\alpha_{0})}}\right),0\right),& \frac{\alpha_{0}}{2}\neq \beta,~y>0, \\\max\left(
\frac{2(\alpha_{2}\sqrt{y}-\alpha_{0}\sqrt{x+y})}{\alpha_{0}\alpha_{2}\sqrt{y}},0\right),& \frac{\alpha_{0}}{2}= \beta,~y>0,\\
0,~\text{otherwise}.
\end{cases}
 \end{equation}
 \end{figure*}
\end{corollary}
The uniform bound derived above and Lemma \ref{Lem:EstimatesPSD}(c) can be used to estimate the deviation of the position and velocity from their desired values as follows:
\begin{corollary} \label{lem:LpLvLf}
Assume that the functions $p$, $v$, $R$, and $\omega$, given by \eqref{eq:pdot}--\eqref{eq:wdot}, with \eqref{eqn:ep}--\eqref{eqn:eW} and \eqref{eqn:f}--\eqref{eq:bd1bd2bd3}, exist over $I$ and that  condition \eqref{eq:Well-DefinednessRequirement} holds over $I$.
In addition, assume  \eqref{eq:BoundingAcceleration}, \eqref{eq:InitialPsi}, and \eqref{eq:InitialEomega} hold.  Let the functions $V_{1}$, $V_{2}$, and $V$ be defined as  in \eqref{eq:V1}--\eqref{eq:V}, where conditions \eqref{eq:Boundc1} and \eqref{eq:Boundc2} hold. Define, for $x,y\in \mathbb{R}_{+}$, 
\begin{align}\label{eq:L_p}
\mathcal{L}_{p}(x,y)&\defas \norm{[\mathrm{I}_{3},0_{3\times 3}]M^{-\frac{1}{2}}_{1}}\mathcal{L}_{u}(x,y),\\ \label{eq:L_v}
\mathcal{L}_{v}(x,y)&\defas \norm{[0_{3\times 3},\mathrm{I}_{3}]M^{-\frac{1}{2}}_{1}}\mathcal{L}_{u}(x,y),\\
\label{eq:L_f}
\mathcal{L}_{f}(x,y)&\defas \norm{[k_{p}\mathrm{I}_{3},k_{v}\mathrm{I}_{3}]M^{-\frac{1}{2}}_{1}}\mathcal{L}_{u}(x,y),
\end{align}
then,  for all $t\in I$,\footnote{As $\mathcal{L}$, given in equation \eqref{eq:L}, is monotonically increasing with respect to its first two arguments, then it follows, using  the definition of $\mathcal{L}_{u}$, given by equation \eqref{eq:L_u}, that  $\mathcal{L}_{u}$ is monotonically increasing with respect to its two arguments. Therefore,  $\mathcal{L}_{p}$, $\mathcal{L}_{u}$, and $\mathcal{L}_{f}$, given by \eqref{eq:L_p}, \eqref{eq:L_v}, and \eqref{eq:L_f}, respectively, are also monotonically increasing with respect to their arguments.} 
    \begin{align*}
\norm{e_{p}(t)}&\leq \mathcal{L}_{p}(V_{1}(0),V_{2}(0)),\\ \norm{e_{v}(t)}&\leq \mathcal{L}_{v}(V_{1}(0),V_{2}(0)),\\   \norm{k_{p}e_{p}(t)+k_{v}e_{v}(t)}&\leq \mathcal{L}_{f}(V_{1}(0),V_{2}(0)).
    \end{align*} 
\end{corollary}

\subsection{Well-definedness of the closed-loop quadrotor system over any arbitrary time interval.}
\label{sec:Well-Definedness}
In  Theorem \ref{thm:LyapunovStability}, we assume that \eqref{eq:Well-DefinednessRequirement} holds and that the closed-loop system is well-defined over $I$. In the theorem below, we will show that, when imposing appropriate constraints on on $p_{d}$, condition \eqref{eq:Well-DefinednessRequirement}  is satisfied and well-definedness is guaranteed.
\begin{theorem}
\label{thm:Well-Definedness}
Let conditions \eqref{eq:BoundingAcceleration}, \eqref{eq:InitialPsi}, and \eqref{eq:InitialEomega} hold. In addition, assume 
\begin{equation}
\label{eq:LowerBoundPddot}m\ddot{p}_{d,3}(t)\geq    \alpha_{f}- mg+\varepsilon,~t\in I,
\end{equation}
for some $\alpha_{f},\varepsilon>0$.   Let $V_{1}$, $V_{2}$, and $V$ be defined as in \eqref{eq:V1}--\eqref{eq:V}, where conditions \eqref{eq:Boundc1} and \eqref{eq:Boundc2} hold.  If $\mathcal{L}_{f}(V_{1}(0),V_{2}(0))$, computed according to \eqref{eq:L_f}, satisfies 
\begin{equation}\label{eq:ConditionOnMathcalF}
\alpha_{f}\geq \mathcal{L}_{f}(V_{1}(0),V_{2}(0)),
\end{equation}
then the closed-loop quadrotor system is well-defined over $I$ in the sense that the functions $p$, $v$, $R$, and $\omega$, given by \eqref{eq:pdot}--\eqref{eq:wdot}, with \eqref{eqn:ep}--\eqref{eqn:eW} and \eqref{eqn:f}--\eqref{eq:bd1bd2bd3}, exist and are differentiable over $I$. Moreover,  condition \eqref{eq:Well-DefinednessRequirement} holds over $I$ and the conclusions  of Theorem \ref{thm:LyapunovStability} follow.
\end{theorem}
\begin{proof}
Recall the definitions of $R_{d}$ and $\omega_{d}$ given by \eqref{eq:Rd}, \eqref{eq:wd}, \eqref{eq:bd1bd2bd3}, and  \eqref{eq:Fd}. Both $R_{d}$ and $\omega_{d}$ depend on $F_{d}$, and  $F_{d}$ depends on $e_{p}$ and $e_{v}$.  Note that $R_{d}$, $\omega_{d}$, and $\dot{\omega}_{d}$ are differentiable with respect to $F_{d}$  whenever $F_{d}$ is not zero and $\norm{F_{d}(\cdot)}+F_{d,3}(\cdot)$ is not zero.  

It follows from condition \eqref{eq:Boundc1}  that  $V_{1}$ given by \eqref{eq:V1} is nonnegative over its domain of definition with the associated matrix $M_{1}$ being positive definite. Therefore, using Lemma \ref{Lem:EstimatesPSD}(c) and the definition of $\mathcal{L}_{f}$ in \eqref{eq:L_f}, 
\begin{align*}
\norm{k_{p}e_{p}(0)+k_{v}e_{v}(0)}\leq& \norm{[k_{p}\mathrm{I}_{3},k_{v}\mathrm{I}_{3}]M^{-\frac{1}{2}}_{1}}\sqrt{V_{1}(0)}\\
&= \mathcal{L}_{f}(V_{1}(0),0)\\
&\leq \mathcal{L}_{f}(V_{1}(0),V_{2}(0))\leq \alpha_{f}.
\end{align*}
Consequently, 
\begin{align*}
F_{d,3}(0)=&-k_{p}e_{p,3}(0)-k_{v}e_{v,3}(t)+mg+m\ddot{p}_{d,3}(0)\\
\geq & -\mathcal{L}_{f}(V_{1}(0),V_{2}(0))+mg\\
&+   \mathcal{L}_{f}(V_{1}(0),V_{2}(0))-mg+\varepsilon=\varepsilon>0,
\end{align*}
implying $ \norm{F_{d}(0)}, \norm{F_{d}(0)}+F_{d,3}(0)>0$.
This shows that the right-hand sides of \eqref{eq:pdot}--\eqref{eq:wdot}, with \eqref{eqn:ep}--\eqref{eqn:eW} and \eqref{eqn:f}--\eqref{eq:bd1bd2bd3},  are continuously differentiable with respect to $t$, $p$, $v$, $R$, and $\omega$  over a local neighborhood of $(p(0),v(0),R(0), \omega(0))$. Existence-uniqueness results for ordinary differential equations\footnote{While $R$ and $R_{d}$  are matrix-valued functions, they can be cast as  vector-valued functions through  vectorization. Then, system \eqref{eq:pdot}--\eqref{eq:wdot} with the thrust and torque laws \eqref{eqn:f} and \eqref{eqn:M} can be viewed as a non-autonomous differential equation on a Euclidean space, and standard results on finite-dimensional differential equations can be applied.} (see, e.g., \cite[Theorem~3.2.1,~Proposition~ 3.2.2.,~p.~82]{david2018ordinary}) then guarantee the existence of a non-singleton interval $\tilde{I}\subset I$ containing zero over which the functions $p$, $v$, $R$, and  $\omega$ are  well-defined and differentiable.

Let $\tilde{I}_{1}\subset I$ be the \textit{maximal} interval of existence over which $p$, $v$, $R$, and $\omega$  are well-defined and differentiable, and define 
\begin{equation}\nonumber
\label{eq:TildeI2}
\tilde{I}_{2}\defas \{t\in \tilde{I}_{1}|F_{d,3}(t)\leq 0\},~
     \tilde{t}\defas \begin{cases} \inf \tilde{I}_2,& \tilde{I}\neq \emptyset,\\
     \infty,&~\text{otherwise},
     \end{cases}
 \end{equation} 
and $\tilde{I}\defas \Set{t\in \tilde{I}_{1}}{0\leq t< \tilde{t}}$. The continuity of $e_p$,  $e_v$, and $\ddot{p}_d$ (and consequently the continuity of $F_{d}$) over $\tilde{I}_{1}$ and the fact that $F_{d,3}(0)>0$ ensure that $\tilde{t}$  is positive. 
In addition, the definition of $\tilde{I}$  ensures that condition \eqref{eq:Well-DefinednessRequirement} holds over $\tilde{I}$. We can then use the conclusions of Theorem \ref{thm:LyapunovStability}, including Corollary \ref{lem:LpLvLf}, over  $\tilde{I}$. Using Corollary $\ref{lem:LpLvLf}$, we have 
\begin{align*}
F_{d,3}(t)=&-k_{p}e_{p,3}(t)-k_{v}e_{v,3}(t)+mg+m\ddot{p}_{d,3}(t)\\
\geq & -\mathcal{L}_{f}(V_{1}(t),V_{2}(0))+mg\\
&+   \mathcal{L}_{f}(V_{1}(t),V_{2}(0))-mg+\varepsilon=\varepsilon>0
\end{align*}
for all $t\in \tilde{I}$. This uniform lower bound on $F_{d,3}$ and the definition of $\tilde{I}$   imply $\tilde{I}=\tilde{I}_{1}$ as the definition of  $\tilde{t}$  implies it is infinite. 
Corollary \ref{lem:LpLvLf} indicates that $e_{p}$ and $e_{v}$ are uniformly bounded over $\tilde{I}$  (we can use theorem \ref{thm:LyapunovStability} to also  show that  $e_{R}$ and $e_{\omega}$ are uniformly bounded over $\tilde{I}$).  Moreover, and by construction,  $p_{d}$ and its first four derivatives are  uniformly bounded. Therefore, using the uniform boundedness of $e_{p}$ and $e_{v}$, we  have $p$ and $v$ uniformly bounded over $\tilde{I}$.  Furthermore, by definition,  $R$ is uniformly bounded ($\norm{R(\cdot)}=1$) over $\tilde{I}$. The function $\omega_{d}$, which depends on $R_{d}$ and its derivative (see \eqref{eq:wd}), is uniformly bounded. Verifying this fact can be outlined as follows. Recalling the definition of $R_{d}$ in \eqref{eq:Rd}, and using the chain rule, we have  $\hat{\omega}_{d}(\cdot)=R_{d}^{\intercal}(\cdot)\dot{R}_{d}(\cdot)=R_{d}^{\intercal}(\cdot)\nabla_{F_{d}}{R_{d}}(\cdot),\dot{F}_{d}(\cdot)$, where $\nabla_{F_{d}}{R_{d}}$ is the derivative of $R_{d}$ with respect to $F_{d}$. By construction, $\norm{R_{d}^{\intercal}(\cdot)}=1$.  Moreover, as   $\varepsilon \leq F_{d,3}(\cdot)\leq \norm{F_{d}(\cdot)}\leq \mathcal{L}_{f}(V_{1}(0),V_{2}(0))+m\norm{a_{\max}}, 
$
the components of $\nabla_{F_{d}}{R_{d}}$ are uniformly bounded. Furthermore, $\dot{F}_{d}(\cdot)$, which is computed according to \eqref{eq:Fd_dot}, is uniformly bounded over $\tilde{I}$  due to the uniform boundedness of $e_p$, $e_{v}$, $\dddot{p}_{d}$, and $\Delta_{f}$. It then follows that $\hat{\omega}_{d}(\cdot)$ is uniformly bounded. Using equation \eqref{eqn:eW} and  the uniform boundedness of $e_{\omega}$, $R$, $R_{d}$, and  $\omega_{d}$, $\omega$ is uniformly bounded over $\tilde{I}$. As  $\tilde{I}$ is the maximal interval of existence within $I$  and that $p(\tilde{I})$, $v(\tilde{I})$, $R(\tilde{I})$, and $\omega(\tilde{I})$ are bounded, it follows that (see, e.g., \cite[Corollary~ 4.10,~p.~111]{logemann2014ordinary}, \cite[Theorem~ 4.1.2,~p.~112]{david2018ordinary}, and  \cite[Theorem~1.4.1,~p.~18]{kong2014short})  $\tilde{I}_{1}=I$ and that completes the proof.    
\end{proof}

In the next corollary, we derive uniform bounds that are independent of  both time and the values of $V_{1}$ and $V_{2}$, where we utilize the monotonicity of $\mathcal{L}_{u}$. Such bounds can then be used to synthesize a robust safe trajectory and an associated neighborhood of initial values. 

We note that if conditions \eqref{eq:InitialPsi} and \eqref{eq:InitialEomega} hold, it follows, where Lemma \ref{lem:BoundingPsiUsingEr} is used, that  
\begin{align*}
V_{2}(0)=&\frac{1}{2}e_{\omega}(0)\cdot J e_{\omega}(0)+k_{R}\Psi(0)+c_{2}e_{R}(0)\cdot e_{\omega}(0)\\
&\leq k_{R}\overline{\Psi}+c_{2}\norm{e_{R}(0)}\norm{e_{\omega}(0)}\\
&\leq k_{R}\overline{\Psi}+c_{2}\sqrt{2\Psi(0)}\norm{J^{-1/2}}\norm{J^{1/2}e_{\omega}(0)}
\leq \overline{\mathcal{V}}_{2},
\end{align*}
where 
\begin{equation}\label{eq:BarV2}
\overline{\mathcal{V}}_{2}\defas \left(k_{R}+2c_{2}\sqrt{\frac{k_{R}}{\underline{\lambda}(J)}\alpha_{\psi}(1-\alpha_{\psi})}\right)\overline{\Psi}.
\end{equation}
\begin{corollary} \label{cor:pvf_bound}
   Let $\alpha_{\psi}\in \intoo{0,1}$, $\overline{\Psi}\in \intoo{0,2}$, and $\overline{\mathcal{V}}_{1}\in \intoo{0,\infty}$ be given and $\bar{\mathcal{V}}_{2}$ is computed according to \eqref{eq:BarV2}. 
Assume there exists a finite time $T>0$ and a four-times continuously differentiable $p_{d}\colon \intcc{0,T}\rightarrow \mathbb{R}^{3}$ satisfying \eqref{eq:BoundingAcceleration}, and  
$
m\ddot{p}_{d,3}(t)\geq \mathcal{L}_{f}(\overline{\mathcal{V}}_{1},\overline{\mathcal{V}}_{2})-mg+\varepsilon,
$
$t\in\intcc{0,T},$ for some $\varepsilon>0$.
Then,  for any initial value $(p(0),v(0),R(0),\omega(0))$ satisfying
\begin{equation} \label{eq:InitialSet}
\begin{split}
 \Psi(0)&\leq \alpha_{\psi}\overline{\Psi},\\
\frac{1}{2}e_{\omega}^{\intercal}(0)Je_\omega(0) &\leq   {k_{R}}(1-\alpha_{\psi})\overline{\Psi},\\
V_{1}(0)&\leq \overline{\mathcal{V}}_{1},
\end{split}
\end{equation}
the functions $p$, $v$, $R$, and $\omega$, given by \eqref{eq:pdot}--\eqref{eq:wdot}, with \eqref{eqn:ep}--\eqref{eqn:eW} and \eqref{eqn:f}--\eqref{eq:bd1bd2bd3}, exist over $\intcc{0,T}$ and condition \eqref{eq:Well-DefinednessRequirement} holds over $\intcc{0,T}$. Moreover,
for all $t\in \intcc{0,T}$,
\begin{align*}
\norm{e_{p}(t)}&\leq \mathcal{L}_{p}(\overline{\mathcal{V}}_{1},\overline{\mathcal{V}}_{2}),\\ 
\norm{e_{v}(t)}&\leq \mathcal{L}_{v}(\overline{\mathcal{V}}_{1},\overline{\mathcal{V}}_{2}),\\   \norm{k_{p}e_{p}(t)+k_{v}e_{v}(t)}&\leq \mathcal{L}_{f}(\overline{\mathcal{V}}_{1},\overline{\mathcal{V}}_{2}),
    \end{align*} 
where $\mathcal{L}_{p},~ \mathcal{L}_{v}$, and $\mathcal{L}_{f}$ are given by \eqref{eq:L_p}, \eqref{eq:L_v}, and \eqref{eq:L_f}, respectively.
\end{corollary}
A direct consequence of the above corollary is the following result which is the base of our trajectory generation procedure. 
\begin{theorem}
\label{thm:ConditionsForSolvingReachAvoidProblem}
    Let $\alpha_{\psi} \in \intoo{0,1}$, $\overline{\Psi}\in \intoo{0,2}$, and $\overline{\mathcal{V}}_{1}\in \intoo{0,\infty}$ be given and $\bar{\mathcal{V}}_{2}$ is computed according to \eqref{eq:BarV2}. 
Assume 
\begin{equation}\label{eq:SatisfyingThrustBound}
m\norm{a_{\max}}+\mathcal{L}_{f}(\overline{\mathcal{V}}_{1},\overline{\mathcal{V}}_{2})\leq f_{\max},
\end{equation}
and that there exists a finite time $T>0$ and a four-times continuously differentiable $p_{d}\colon \intcc{0,T}\rightarrow \mathbb{R}^{3}$ satisfying \eqref{eq:BoundingAcceleration}, and  
    \begin{equation}
        \begin{split} \label{eq:PropertiesOfDesiredTrajectory}
        p_{d}(t)&\in (\mathcal{X}_{\mathrm{s}}-\mathcal{L}_{p}(\overline{\mathcal{V}}_{1},\overline{\mathcal{V}}_{2})\mathbb{B}_{3})\setminus (\mathcal{X}_{\mathrm{u}}+\mathcal{L}_{p}(\overline{\mathcal{V}}_{1},\overline{\mathcal{V}}_{2})\mathbb{B}_{3}),\\
        \abs{\dot{p}_{d}(t)}&\leq v_{\max}-\mathcal{L}_{v}(\overline{\mathcal{V}}_{1},\overline{\mathcal{V}}_{2})1_{3},\\  
        m\ddot{p}_{d,3}(t)&\geq \mathcal{L}_{f}(\overline{\mathcal{V}}_{1},\overline{\mathcal{V}}_{2})-mg+\varepsilon,
        \\
         p_{d}(T)&\in \mathcal{X}_{\mathrm{t}}-\mathcal{L}_{p}(\overline{\mathcal{V}}_{1},\overline{\mathcal{V}}_{2})\mathbb{B}_{3},
        \end{split}
    \end{equation}
$t\in \intcc{0,T}$, for some $\varepsilon>0$, and $\mathcal{L}_{p},~ \mathcal{L}_{v}$, and $\mathcal{L}_{f}$ are given by \eqref{eq:L_p}, \eqref{eq:L_v}, and \eqref{eq:L_f}, respectively. Then,  for any initial value $(p(0),v(0),R(0),\omega(0))$ satisfying \eqref{eq:InitialSet}, the functions $p$, $v$, $R$, and $\omega$, given by \eqref{eq:pdot}--\eqref{eq:wdot}, with \eqref{eqn:ep}--\eqref{eqn:eW} and \eqref{eqn:f}--\eqref{eq:bd1bd2bd3}, exist over $\intcc{0,T}$ and condition \eqref{eq:Well-DefinednessRequirement} holds over $\intcc{0,T}$. Moreover,  the resulting position  $p$ and thrust $f$  satisfy 
\eqref{eq:ReachAVoidProblem}.
\end{theorem}

It is essential to show that the nominal point $(p_{0},v_{0},R_{0},\omega_{0})$ satisfies \eqref{eq:InitialSet}. The result below imposes conditions on $p_{d}$ to ensure that the nominal point satisfies \eqref{eq:InitialSet}.  
\begin{lemma}[See the proof in the Appendix]\label{lem:ConditionsOnPd}
If 
\begin{equation}\label{eq:InitialConditionsPd}
\begin{split} 
p_{d}(0)=p_{0},~\,\dot{p}_{d}(0)=v_{0},~
\ddot{p}_{d}(0)=0_{3},~\,\dddot{p}_{d}(0)=0_{3}.
\end{split}
\end{equation}
Then, for $(p(0),v(0),R(0),\omega(0))=(p_{0},v_{0},R_{0},\omega_{0})$, $R(0)=R_{d}(0)$, and
$\omega(0)=\omega_{d}(0).$
\end{lemma}
\begin{remark}
In order to satisfy \eqref{eq:InitialConditionsPd} while having the  first three conditions of \eqref{eq:PropertiesOfDesiredTrajectory} holding at $t=0$, we  need the choices of  $\alpha_{\psi} \in \intoo{0,1}$, $\overline{\Psi}\in \intoo{0,2}$, and $\overline{\mathcal{V}}_{1}\in \intoo{0,\infty}$ to ensure that 
\begin{align*}
p_{0}&\in (\mathcal{X}_{\mathrm{s}}-\mathcal{L}_{p}(\overline{\mathcal{V}}_{1},\overline{\mathcal{V}}_{2})\mathbb{B}_{3})\setminus (\mathcal{X}_{\mathrm{u}}+\mathcal{L}_{p}(\overline{\mathcal{V}}_{1},\overline{\mathcal{V}}_{2})\mathbb{B}_{3}),\\
\abs{v_{0}}&
\leq v_{\max}-\mathcal{L}_{v}(\overline{\mathcal{V}}_{1},\overline{\mathcal{V}}_{2})1_{3},\\   
0& \geq \mathcal{L}_{f}(\overline{\mathcal{V}}_{1},\overline{\mathcal{V}}_{2})-mg+\varepsilon.
\end{align*}
\end{remark}

The result above indicates that if $(p(0),v(0),R(0),\omega(0))=(p_{0},v_{0},R_{0},\omega_{0})$, where $p_{d}$ satisfies \eqref{eq:InitialConditionsPd}, then we have $\Psi(0)=\norm{e_{\omega}(0)}=V_{1}(0)=0$. 
Using  a  continuity argument, we have:
\begin{corollary}
Let $\alpha_{\psi}, \in \intoo{0,1}$, $\overline{\Psi}\in \intoo{0,2}$, and $\overline{\mathcal{V}}_{1}\in \intoo{0,\infty}$ be given and $\bar{\mathcal{V}}_{2}$ is computed according to \eqref{eq:BarV2}.      Assume  $p_{d}$ satisfies    \eqref{eq:BoundingAcceleration}, \eqref{eq:PropertiesOfDesiredTrajectory}, and \eqref{eq:InitialConditionsPd}, then there exists a neighborhood of $(p_{0},v_{0},R_{0},\omega_{0})$, relative to $\mathbb{R}^{3}\times \mathbb{R}^{3}\times \SO\times \mathbb{R}^{3}$,  such that for any initial value in that neighborhood, the conditions in  \eqref{eq:InitialSet} hold.
\end{corollary}

\section{Trajectory Generation}
\label{sec:TrajectoryGeneration}
 Let $\alpha_{\psi}, \in \intoo{0,1}$, $\overline{\Psi}\in \intoo{0,2}$, and $\overline{\mathcal{V}}_{1}\in \intoo{0,\infty}$ be given, $\bar{\mathcal{V}}_{2}$ be computed according to \eqref{eq:BarV2}, and $\mathcal{L}_{p}(\overline{\mathcal{V}}_{1},\overline{\mathcal{V}}_{2})$, $\mathcal{L}_{v}(\overline{\mathcal{V}}_{1},\overline{\mathcal{V}}_{2})$, and $\mathcal{L}_{f}(\overline{\mathcal{V}}_{1},\overline{\mathcal{V}}_{2})$ be computed according to \eqref{eq:L_p}, \eqref{eq:L_v}, and \eqref{eq:L_f}, respectively. In this section, we illustrate how a desired trajectory $p_{d}$ is computed. Our method combines some of the well-established  planning approaches \cite{marcucci2023motion,kavraki1996probabilistic} with efficient hyper-rectangular computations \cite{Serry2024Safe}. Informally speaking, our approach relies on a computing a safe tube  of connected  hyper-rectangles that do not intersect with the unsafe set, then conducting an optimization procedure that results in a polynomial trajectory that passes through the safe tube. Below, we illustrate how the safe tube is obtained.
\subsection{Computing a safe tube}
In view of Theorem \ref{thm:ConditionsForSolvingReachAvoidProblem}, define 
\begin{align*}
\tilde{\mathcal{X}}_{\mathrm{o}}&\defas \mathcal{X}_{\mathrm{o}}-\mathcal{L}_{p}(\overline{\mathcal{V}}_{1},\overline{\mathcal{V}}_{2}) \mathbb{B}_{3}^{\infty},\\
\tilde{\mathcal{X}}_{\mathrm{u}}^{(i)}&\defas \Hintcc{\underline{x}_{\mathrm{u}}^{(i)},\overline{x}_{\mathrm{u}}^{(i)}}+\mathcal{L}_{p}(\overline{\mathcal{V}}_{1},\overline{\mathcal{V}}_{2}) \mathbb{B}_{3}^{\infty},~i\in \intcc{1;N_{u}},\\
\tilde{\mathcal{X}}_{\mathrm{u}}&\defas \bigcup_{i=1}^{N_{\mathrm{u}}}\tilde{\mathcal{X}}_{\mathrm{u}}^{(i)},\\
\tilde{\mathcal{X}}_{\mathrm{t}}&\defas {\mathcal{X}}_{\mathrm{t}}-\mathcal{L}_{p}(\overline{\mathcal{V}}_{1},\overline{\mathcal{V}}_{2}) \mathbb{B}_{3}^{\infty}.
\end{align*}
The sets above correspond to an inflated version of the unsafe set ($\tilde{\mathcal{X}}_{\mathrm{u}}$) and deflated versions of the operating domain and target set ($\tilde{\mathcal{X}}_{\mathrm{o}}$ and $\tilde{\mathcal{X}}_{\mathrm{t}}$), where deflation and inflation are based on the position error bound $\mathcal{L}_{p}(\overline{\mathcal{V}}_{1},\overline{\mathcal{V}}_{2})$ defined by \eqref{eq:L_p}. Note that in our deflation and inflation procedure, we use the unit ball $\mathbb{B}_{3}^{\infty}$ instead of $\mathbb{B}_{3}$, which is more conservative yet computationally more efficient as Minkowski sums and differences of hyper-rectangles can be computed exactly by summing/subtracting  centers and  radii.  
The safe tube is obtained by  generating $N_{s}+1$ waypoints $\mathtt{p}_{0},\mathtt{p}_{1},\ldots, \mathtt{p}_{N_{s}}\in \mathbb{R}^{3}$, and associated  vector radii $\mathtt{r}_{0},\mathtt{r}_{1},\mathtt{r}_{2},\ldots, \mathtt{r}_{N_{s}}\in \mathbb{R}_{+}^{3}$, satisfying
\begin{equation}\label{eq:DesiredPropertiesOfWayPoints}
\begin{split}
\mathtt{p}_{0}&=p_{0},\\
\mathtt{p}_{i+1}&\in \mathtt{p}_{i}+\Hintcc{-\mathtt{r}_{i},\mathtt{r}_{i}},~i\in \intcc{0;N_{s}-1},\\
\mathtt{p}_{i}+\Hintcc{-\mathtt{r}_{i},\mathtt{r}_{i}}&\subseteq \tilde{\mathcal{X}}_{\mathrm{o}}\setminus \tilde{\mathcal{X}}_{\mathrm{u}} ,~i\in \intcc{0;N_{s}-1},\\
\mathtt{p}_{N_{s}}+\Hintcc{-\mathtt{r}_{N_{s}},\mathtt{r}_{N_{s}}}&\subseteq  \tilde{\mathcal{X}}_{\mathrm{t}}.
\end{split}
\end{equation}
Note that if a desired trajectory $p_{d}\colon \intcc{0,T}\rightarrow \mathbb{R}^{3}$ has values contained in the safe tube, with $p_{d}(T)$ contained in $\mathtt{p}_{N_{s}}+\Hintcc{-\mathtt{r}_{N_{s}},\mathtt{r}_{N_{s}}}$, the first and last conditions of \eqref{eq:PropertiesOfDesiredTrajectory} hold.

The waypoints and the associated safe radii can be estimated by integrating  sampling-based approaches \cite{kavraki1996probabilistic}, with  safe and efficient  hyper-rectangular set-based computations \cite{Serry2024Safe}.  The following technical lemmas illustrate   how  safe hyper-rectangular sets within the operating domain can be computed.
\begin{lemma}[See the proof in the Appendix]\label{lem:SafeBox}
    Let $v\in\Hintcc{a,b}\subseteq \mathbb{R}^{n}$, where $a,b\in \mathbb{R}^{n}$, $a\leq b$, then
    $\mathcal{H}(v,\Hintcc{a,b}) \subseteq \Hintcc{a,b}$, where 
    \begin{equation} \nonumber
\mathcal{H}(v,\Hintcc{a,b})\defas v+\Hintcc{-r,r},
    \end{equation}
 and $
    r=\mathrm{radius}(\Hintcc{a,b})-\abs{\mathrm{center}(\Hintcc{a,b})-v}
    $.   
\end{lemma}

\begin{lemma}[See the proof in the Appendix]\label{lem:SafeStrip}
Let $x\in \mathbb{R}^{n}$ and $\Hintcc{a,b}\subseteq \mathbb{R}^{n}$, where $a,b\in \mathbb{R}^{n}$, $a\leq b$. Moreover, let $r=\mathrm{radius}(\Hintcc{a,b})$ and $c=\mathrm{center}(\Hintcc{a,b})$. Then, $\min_{y\in \Hintcc{a,b}}\norm{x-y}_{\infty}=\norm{x-y^{\ast}}_{\infty}$, where $y^{\ast}=\mathrm{ClosestPoint}(x,\Hintcc{a,b})$,   
\begin{equation}\label{eq:ClosestPoint}\nonumber
(\mathrm{ClosestPoint}(x,\Hintcc{a,b}))_{i}\defas \begin{cases}
  x_{i},~ x_{i}\in \intcc{a_{i},b_{i}},\\
  c_{i}+r_{i}\mathrm{sgn}(x_{i}-c_{i}),~\text{otherwise},
\end{cases}
\end{equation}
$i\in \intcc{1;n}$, and $\mathrm{sgn}(\cdot)$ is the signum function. Now, assume $x\not\in\Hintcc{a,b}$,  
 let $\tilde{i}\in \intcc{1;p}$ such that $\abs{x_{\tilde{i}}-y^{\ast}_{\tilde{i}}}=\norm{x-y^{\ast}}_{\infty}>0$, and define
\begin{equation}\label{eq:SafeStrip} \nonumber
\mathcal{S}(x,\Hintcc{a,b},\alpha)\defas\Set{z\in \mathbb{R}^{n}}{\abs{z_{\tilde{i}}-x_{\tilde{i}}}\leq \alpha \norm{x-y^{\ast}}_{\infty}},
\end{equation}
where $\alpha \in \intco{0,1}$.
Then, $\mathcal{S}(x,\Hintcc{a,b},\alpha )\cap \Hintcc{a,b} = \emptyset$. 
\end{lemma}

From Lemmas \ref{lem:SafeBox} and \ref{lem:SafeStrip}, we deduce:
\begin{corollary}
    \label{lem:SafeBox2}
    Let $y\in \tilde{\mathcal{X}}_{\mathrm{o}}\setminus\tilde{\mathcal{X}}_{\mathrm{u}}$, fix $\alpha \in \intco{0,1}$, and define 
    \begin{equation}\nonumber \mathfrak{R}(y,\tilde{\mathcal{X}}_{\mathrm{o}},\tilde{\mathcal{X}}_{\mathrm{u}},\alpha )\defas\mathcal{H}(y,\tilde{\mathcal{X}}_{\mathrm{o}}) \bigcap \left(\cap_{i=1}^{N_\mathrm{u}}\mathcal{S}(y,\tilde{\mathcal{X}}_{\mathrm{u}}^{(i)},\alpha)\right).
    \end{equation}
    Then, $\mathfrak{R}(y,\tilde{\mathcal{X}}_{\mathrm{o}},\tilde{\mathcal{X}}_{\mathrm{u}},\alpha )\subseteq \tilde{\mathcal{X}_{\mathrm{o}}}\setminus \tilde{\mathcal{X}}_{\mathrm{u}}$.
\end{corollary}
 Using the safe hyper-rectangles in Lemma \ref{lem:SafeBox} and Corollary \ref{lem:SafeBox2},  
 we  construct a  rapidly exploring random tree (RRT) $\mathscr{T}=(\mathscr{V},\mathscr{E})$, with a set of vertices $\mathscr{V}$ and a set of edges $\mathscr{E}$, such that each vertex (except the first one) is contained in a safe hyper-rectangular neighbor of the parent vertex. The metric  used in constructing the tree is as follows: the distance between a sample point $x_{s}$ and a point $x_{v}$ in the vertex set $\mathscr{V}$ is the distance between the sample point and a safe hyper-rectangle containing $x_{v}$, computed according to Lemma \ref{lem:SafeBox2}. If $x_{v}$ happens to be associated with the shortest distance (among all the vertices of the tree), the new point added to the tree is a point within the safe hyper-rectangle that corresponds to that shortest distance. 
 
 Let $N_{v}$ be the maximum number of vertices,   $C_{\mathrm{sample}}\in \intoc{0,1}$ be a parameter that determines the percentage of points to be sampled from $\tilde{\mathcal{X}}_{\mathrm{o}}$ and $\tilde{\mathcal{X}}_{\mathrm{t}}$, and $\mathtt{sample}$ be a sampling function such that $\mathtt{sample}(S)$  randomly generates a point from the set $S$. The random tree is then computed according to Algorithm \ref{Alg:RRTAlgorithm}, where the parameter $\alpha$ used in constructing the safe hyper-rectangles is user-defined.   
\begin{algorithm}
\SetAlgoLined
  $i \gets 1$, $x_{i}\gets p_{0}$,
 $\mathscr{V} \gets \{x_{i}\}$, $\mathscr{E}\gets \emptyset$

 \While{$i\leq  N_{v}$}{

 \eIf{$i\leq C_{\mathrm{sample}} N_{v}$}
  {$x_{s}\gets \mathtt{sample}(\tilde{\mathcal{X}}_{\mathrm{s}}\setminus\tilde{\mathcal{X}}_{\mathrm{u}})$}{$x_{s}\gets \mathtt{sample}(\tilde{\mathcal{X}}_{\mathrm{t}})$}

{$d \gets \infty$,  $j\gets 1$}

  \While{$j\leq \mathrm{card}(\mathscr{V})$}{
  $d'\gets \norm{x_{s}-\mathrm{ClosestPoint}(x_{s},\mathfrak{R}(x_{j},\tilde{\mathcal{X}}_{\mathrm{o}},\tilde{\mathcal{X}}_{\mathrm{u}},\alpha ))}_{\infty}$
 
  \If{ $d'<d$}{
  $d\gets d'$, $i_{\mathrm{near}}\gets j$
  }

$j\gets j+1$
  
  }

   $i\gets i+1$

${x}_{i}\gets \mathrm{ClosestPoint}(x_{\mathrm{s}},\mathfrak{R}({x}_{i_{\mathrm{near}}},\tilde{\mathcal{X}}_{\mathrm{o}},\tilde{\mathcal{X}}_{\mathrm{u}},\alpha ))$

$\mathscr{V}\gets \mathscr{V}\cup \{x_{i}\}$, $\mathscr{E}\gets \mathscr{E}\cup \{(x_{i_{\mathrm{near}}},x_{i})\}$

  \If{$x_{i}\in \tilde{\mathcal{X}}_{\mathrm{t}}$}{
   break
  }
  
 }
\KwResult{$\mathscr{T}=(\mathscr{V},\mathscr{E})$}
 \caption{Constructing an RRT}
 \label{Alg:RRTAlgorithm}
 \end{algorithm}
Once the tree $\mathscr{T}$ is constructed, and assuming there exists a point $x_{t}\in \mathscr{V}$ such that $x_{t}\in \tilde{\mathcal{X}}_{\mathrm{t}}$, a shortest path algorithm can then be conducted over $\mathscr{T}$, connecting $p_{0}$ and $x_{t}$, and resulting in $N_{s}+1$ points $\mathtt{p}_{0}=p_{0},\mathtt{p}_{1},\ldots, \mathtt{p}_{N_{s}}=x_{t}$. After that, the vector radii for the safe hyper-rectangles can be computed as follows:
\begin{equation}
\label{eq:RadiiForSafeTube}
\begin{split}
\mathtt{r}_{i}&=\mathrm{radius}(\mathfrak{R}(\mathtt{p}_{i},\tilde{\mathcal{X}}_{\mathrm{o}},\tilde{\mathcal{X}}_{\mathrm{u}},\alpha)),~i\in \intcc{0;N_{s}-1},\\
\mathtt{r}_{N_{s}}&=\mathrm{radius}(\mathcal{H}(\mathtt{p}_{N_{s}},\tilde{\mathcal{X}_{\mathrm{t}}})).
\end{split}
\end{equation}
By construction, the points  $\mathtt{p}_{0},\mathtt{p}_{1},\ldots, \mathtt{p}_{N_{s}}\in \mathbb{R}^{3}$, and associated vector radii $\mathtt{r}_{0},\mathtt{r}_{1},\mathtt{r}_{2},\ldots, \mathtt{r}_{N_{s}}\in \mathbb{R}_{+}^{3}$ satisfy \eqref{eq:DesiredPropertiesOfWayPoints}.

\subsection{Piecewise B\'ezier curve}
Once the safe tube is obtained, we compute the desired trajectory using B\'ezier curves. We assume $p_{d}$ to be a piecewise continuous function with $N_{s}$ segments, where each segments consists of $N_{p}$+1 points. The parameter $N_{p}$ is user-defined. Let $\delta_{i},~i\in \intcc{1;N_{s}},$ be the duration of each segment, and define 
$
t_{0}=0,~t_i=t_{i-1}+\delta_{i},~i\in \intcc{1;N_{s}},~T=\sum_{i=1}^{N_{s}}\delta_{i}.
$
Then, $p_{d}\colon \intcc{0,T}\rightarrow \mathbb{R}^{3}$ has the form (see Section \ref{sec:BezierCurveDef})
\begin{equation}\label{eq:PWBezier}
p_{d}(t)=\begin{cases}
    \sum_{i=0}^{N_p}c_{1}^{i}\mathfrak{b}_{i,n}(\frac{t-t_{0}}{\delta_{1}}),~t\in \intcc{t_{0},t_{1}},\\
       \sum_{i=0}^{N_p}c_{2}^{i}\mathfrak{b}_{i,n}(\frac{t-t_{1}}{\delta_{2}}),~t\in \intcc{t_{1},t_{2}},\\
       \vdots\\
    \sum_{i=0}^{N_p}c_{N_{s}}^{i}\mathfrak{b}_{i,n}(\frac{t-t_{N_{s}-1}}{\delta_{N_{s}}}),~t\in \intcc{t_{N_{s}-1},t_{N_{s}}}.
\end{cases}
\end{equation}
The control points $c_{i}^{j},~i\in \intcc{1;N_{s}},~j\in \intcc{0;N_{p}}$, are required to satisfy the following set of constraints, which are based on Theorem \ref{thm:ConditionsForSolvingReachAvoidProblem} and Lemma \ref{lem:ConditionsOnPd}\footnote{Additional constraints can be imposed if necessary such as requiring $\dot{p}_{d}(T)=\ddot{p}_{d}(T)=0_{3}$.}:
\begin{itemize}
\item the initial value of $p_{d}$ is $p_{0}$:
    \begin{equation}\label{eq:Constraint1}
        c_{1}^{0}=\mathtt{p}_{0},
    \end{equation}  
 \item  the initial value of $\dot{p}_{d}$ is $v_{0}$:
   \begin{equation}\label{eq:Constraint2}
    \frac{N_{p}}{\delta_{1}}(c_{1}^{1}-c_{1}^{0})=v_{0},
    \end{equation} 
 \item  the initial value of $\ddot{p}_{d}$ is $0_{3}$:
    \begin{equation}\label{eq:Constraint3}
c_{1}^{2}-2c_{1}^{1}+c_{1}^{0}=0_{3},
    \end{equation}
    \item  the initial value of $\dddot{p}_{d}$ is $0_{3}$:
   \begin{equation}\label{eq:Constraint4}
 c_{1}^{3}-3c_{1}^{2}+3c_{1}^{1}-c_{1}^{0}=0_{3},
    \end{equation}
 \item   safety of the generated trajectory in the sense that $p_{d}(t)\in \mathtt{p}_{i-1}+\Hintcc{-\mathtt{r}_{i-1},\mathtt{r}_{i-1}},~t\in \intcc{t_{i-1},t_{i}},~i\in \intcc{1;N_{s}}$: 
\begin{equation}\label{eq:Constraint5}
c_{i}^{j}\in \mathtt{p}_{i-1}+\Hintcc{-\mathtt{r}_{i-1},\mathtt{r}_{i-1}}~\forall j\in \intcc{0;N_{p}},~i\in \intcc{1;N_{s}},
\end{equation}
\item continuity of $p_{d}$ at the junction points: 
\begin{equation}\label{eq:Constraint6}
c_{i}^{N_{p}}=c_{i+1}^{0} ~\forall i\in \intcc{1;N_{s}-1},
\end{equation}
\item satisfying the velocity bound $\abs{\dot{p}_{d}(t)}\leq v_{\max}-\mathcal{L}_{v}(\overline{\mathcal{V}}_{1},\overline{\mathcal{V}}_{2})1_{3}$, $t\in \intcc{0,T}$:
\begin{equation}\label{eq:Constraint7}
\begin{split}
  \abs{\frac{N_{p}}{\delta_{i}}(c_{i}^{j+1}-c_{i}^{j})}\leq v_{\max}-\mathcal{L}_{v}(\overline{\mathcal{V}}_{1},\overline{\mathcal{V}}_{2})1_{3},\\
  ~j\in \intcc{0;N_{p}-1},~i\in \intcc{1;N_{s}},
  \end{split}
\end{equation}
\item continuity of $\dot{p}_{d}$ at the junction points: 
\begin{equation}\label{eq:Constraint8}
\frac{1}{\delta_{i}}(c_{i}^{N_{p}}-c_{i}^{N_{p}-1})=\frac{1}{\delta_{i+1}}(c_{i+1}^{1}-c_{i+1}^{0})  ~\forall i\in \intcc{1;N_{s}-1},
\end{equation}
\item  satisfying the bound $\abs{ge_{3}+\ddot{p}_{d}(t)}\leq a_{\max},~t\in \intcc{0,T}$:
\begin{equation}
\label{eq:Constraint9}
\begin{split}
\abs{\frac{N_{p}(N_{p}-1)}{\delta_{i}^{2}}(c_{i}^{j}-2c_{i}^{j+1}+c_{i}^{j+2})+ge_{3}}&\leq\\ {a_{\max}},~j\in \intcc{0;N_{p}-2},~i\in \intcc{1;N_{s}},&
\end{split}
\end{equation}
\item satisfying the condition $m\ddot{p}_{d,3}(t)\geq     \mathcal{L}_{f}(\overline{\mathcal{V}}_{1},\overline{\mathcal{V}}_{2})-mg+\varepsilon,~t\in \intcc{0,T}$:
\begin{equation}
\label{eq:Constraint10}
\begin{split}
\frac{m N_{p}(N_{p}-1)}{\delta_{i}^{2}}(c_{i,3}^{j}-2c_{i,3}^{j+1}+c_{i,3}^{j+2})\geq \\ {\mathcal{L}_{f}(\overline{\mathcal{V}}_{1},\overline{\mathcal{V}}_{2})}-mg+\varepsilon,~j\in \intcc{0;N_{p}-2},~i\in \intcc{1;N_{s}},&
\end{split}
\end{equation}
\item continuity of $\ddot{p}_{d}$ at the junction points: 
\begin{equation}
\label{eq:Constraint11}
\begin{split}
\frac{1}{\delta_{i}^{2}}(c_{i}^{N_{p}}-2c_{i}^{N_{p}-1}+c_{i}^{N_{p}-2})&=\\
\frac{1}{\delta_{i+1}^{2}}(c_{i+1}^{2}-2c_{i+1}^{1}+c_{i+1}^{0}),&  ~ i\in \intcc{1;N_{s}-1},
\end{split}
\end{equation}
\item continuity of $\dddot{p}_{d}$ at the junction points: 
\begin{equation}
\label{eq:Constraint12}
\begin{split}
\frac{1}{\delta_{i}^{3}}(c_{i}^{N_{p}}-3c_{i}^{N_{p}-1}+3c_{i}^{N_{p}-2}-c_{i}^{N_{p}-3})&=\\
\frac{1}{\delta_{i+1}^{3}}(c_{i+1}^{3}-3c_{i+1}^{2}+3c_{i+1}^{1}-c_{i+1}^{0}),&  ~ i\in \intcc{0;N_{s}-1},
\end{split}
\end{equation}
\item continuity of $\ddddot{p}_{d}$ at the junction points: 
\begin{equation}
\label{eq:Constraint13}
\begin{split}
\frac{1}{\delta_{i}^{4}}(c_{i}^{N_{p}}-4c_{i}^{N_{p}-1}+6c_{i}^{N_{p}-2}-4c_{i}^{N_{p}-3}+c_{i}^{N_{p}-4})&=\\
\frac{1}{\delta_{i+1}^{4}}(c_{i+1}^{4}-
4c_{i+1}^{3}+6c_{i+1}^{2}-4c_{i+1}^{1}+c_{i+1}^{0}),\\
~ i\in \intcc{0;N_{s}-1},
\end{split}
\end{equation}
\item the  value of $p_{d}$ at final time  is in $\mathtt{p}_{N_{s}}+\Hintcc{-\mathtt{r}_{N_{s}},\mathtt{r}_{N_{s}}}$:
\begin{equation}
\label{eq:Constraint14}
\mathtt{p}_{N_{s}}-\mathtt{r}_{N_{s}}\leq c^{N_{p}}_{N_{s}}\leq \mathtt{p}_{N_{s}}+\mathtt{r}_{N_{s}}.
\end{equation}
\end{itemize}
Note that if  the time durations $\delta_{1},\cdots, \delta_{N_{s}}$ are specified, the constraints \eqref{eq:Constraint1}--\eqref{eq:Constraint14} above are linear with respect to the control points. In Algorithm \ref{Alg:DesiredTrajectory}, we present a heuristic approach to determine  the desired trajectory $p_{d}$  by means of iterative linear programming. The heuristic approach relies on initially guessing the  time $T$ and incrementally increasing it until the constraints become feasible, where each duration $\delta_{i}$ is assumed to be a fraction of $T$ that depends on the ratio of the distance between the consecutive two waypoints $\mathtt{p}_{i-1}$ and $\mathtt{p}_{i}$ over the total length of the piecewise linear curve connecting all the waypoints. We assume the value of $N_{p}$ is given and fixed  during the synthesis procedure.\footnote{The trajectory resulting from Algorithm \ref{Alg:DesiredTrajectory} may be  sub-optimal (e.g., in terms of  jerk or snap). If optimality criteria are important to impose on the desired trajectory, then nonlinear optimization methods may be adopted (see, e.g., \cite{marcucci2023motion,gao2018online, richter2016polynomial}).} 
\begin{algorithm}
 \SetAlgoLined  
 {Define 
$l_{i}=\norm{\mathtt{p}_{i}-\mathtt{p}_{i-1}},~i\in \intcc{1;N_{s}}$,
$L=\sum_{i=1}^{N_{s}} l_{i}$, 
$q_{i}=\frac{l_{i}}{L},~i\in \intcc{1;N_{s}}$.
}
 
 {Let $T_{0}$ be a positive parameter specifying an initial guess for the full time horizon and $\alpha_{t}>1$ be a scaling  parameter}.

{Define
$
\delta_{i}=q_{i}T_{0},~i\in \intcc{1;N_{s}}.
$
\label{Step:DefiningTaus}}

{With the values of $\delta_{i},~i\in \intcc{1;N_{s}},$ defined in step \ref{Step:DefiningTaus}, solve a linear program involving the control points $c_{i}^{j},~j\in \intcc{0;N_{p}},~i\in \intcc{1;N_{s}}$, while considering the constraints \eqref{eq:Constraint1}--\eqref{eq:Constraint14}.\label{Step:LinearProgram}} 

{If the linear  program in step \ref{Step:LinearProgram} is feasible,  the resulting control points in addition to the time durations $\delta_{i},~i\in \intcc{1;N_{s}}$,  can then be used to obtain the desired trajectory according to equation \eqref{eq:PWBezier}.}

{If the linear  program in step \ref{Step:LinearProgram} is infeasible, redefine $T_{0}$ as
$T_{0}\defas \alpha_{T}T_{0},
$
and repeat steps \ref{Step:DefiningTaus} and \ref{Step:LinearProgram}.}

\caption{Computing the desired trajectory}
\label{Alg:DesiredTrajectory}
\end{algorithm}

If the procedure in Algorithm \ref{Alg:DesiredTrajectory} is successful, we have the resulting trajectory 
satisfying \eqref{eq:BoundingAcceleration}, \eqref{eq:PropertiesOfDesiredTrajectory}, and \eqref{eq:InitialConditionsPd}. Therefore, the reach-avoid problem in this work is solved successfully, where the safe initial set containing $(p_{0},v_{0},R_{0},\omega_{0})$ is characterized by \eqref{eq:InitialSet}.

\section{Numerical Simulations}
\label{sec:Simulations}
Numerical quadrotor simulations (using MATLAB) incorporating the  proposed control synthesis approach  are conducted to demonstrate its performance and effectiveness. For the purpose of consistency, all calculations and computations are done on a computer with an i7-12700 CPU. The  hyper-rectangular plots presented in this section  are obtained  with the MATLAB command $\mathtt{plotcube}$ \cite{Oliver2024}.

\begin{figure}[ht]
    \centering
\includegraphics[width=0.99\columnwidth]{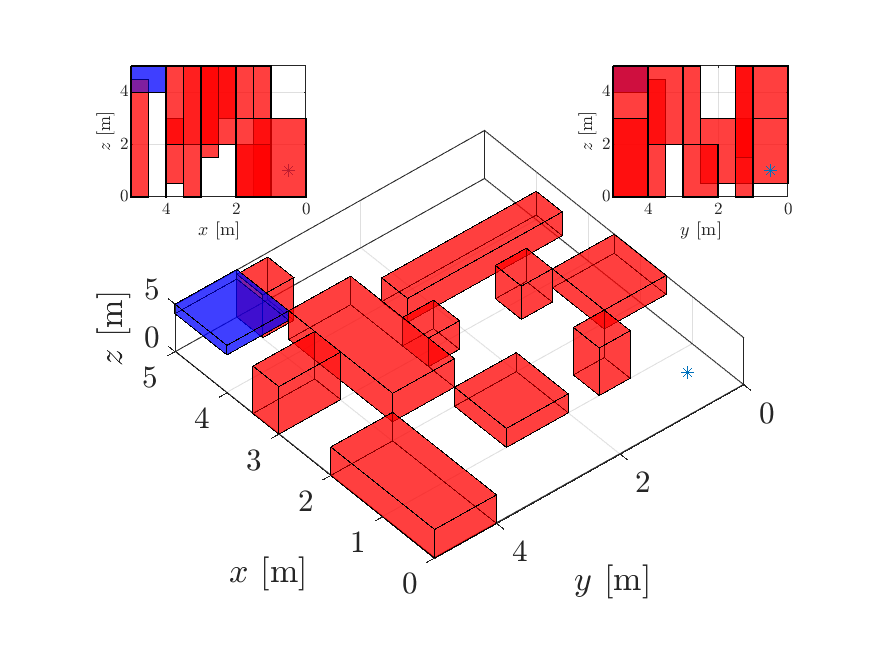}
    \caption{The environment with ten obstacles as red boxes, one target set as the blue box, and the starting point as the blue asterisk. The top left corner is the view along the negative direction of the y-axis, and the top right corner is the view along the positive direction of the x-axis.}
    \label{fig:environment}
\end{figure}
The parameters of the quadrotor are adopted from \cite{lee2010geometric}, where 
$$J = \diag([0.0820, 0.0845, 0.1377]^{\intercal})~\mathrm{kg}\cdot\mathrm{m}^2,~m = 4.34~\mathrm{kg}.$$

\subsection{Reach-avoid problem setup}
We consider a reach-avoid control scenario, where the operating domain is defined as a cube with edges measuring 5 meters each ($\mathcal{X}_{\mathrm{o}} = \Hintcc{[0, 0, 0]^{\intercal}, [5, 5,  5]^{\intercal}}$), the target set is given by $\mathcal{X}_{\mathrm{t}} = \Hintcc{[4, 4, 4]^{\intercal}, [5, 5, 5]^{\intercal}}$, and the unsafe set is given as a union of ten hyper-rectangles. The initial nominal position of the  quadrotor is $p_{0}=[0.5, 0.5, 1]^{\intercal}$ and the initial nominal velocity is $v_{0}=0_{3}$.   The operating domain, the target and the unsafe sets, and the initial nominal position are depicted in Figure \ref{fig:environment}. We let $v_{\max} = [2, 2, 2]^{\intercal}$,  $f_{\max}=2mg=85.1508$ N, and $a_{\max} = [1, 1, 10]^{\intercal}$. The control synthesis process starts by setting   $\overline{\Psi} = 0.005, ~\alpha_\psi = 0.4, ~\overline{\mathcal{V}}_{1} = 0.4$. 

\subsection{Gain tuning through optimization}
While our theoretical results imply local exponential stability of the closed-loop quadrotor dynamics for any choice of positive control gains, that choice should ensure that the theoretical uniform bounds, which depend implicitly on the control gains, are not too conservative. In particular, we want the gains choice to result in small values for the uniform bounds $\mathcal{L}_{p}(\overline{\mathcal{V}}_{1},\overline{\mathcal{V}}_{2})$, $\mathcal{L}_{v}(\overline{\mathcal{V}}_{1},\overline{\mathcal{V}}_{2})$, and $\mathcal{L}_{f}(\overline{\mathcal{V}}_{1},\overline{\mathcal{V}}_{2})$, where $\overline{\mathcal{V}}_{2}$ is computed according to \eqref{eq:BarV2}. Let $\gamma_{1}, \gamma_{2}\in \intoo{0,1}$ be parameters that determine the values of $c_{1}$ and $c_{2}$, used in the definitions of the functions $V_{1}$ and $V_{2}$ in  \eqref{eq:V1} and \eqref{eq:V2}, respectively, through the relations  
\begin{align*} 
 c_{1} & = \gamma_1 \min\left(\sqrt{k_{p}m}, \frac{4mk_{p}k_{v}}{k_{v}^2+4mk_{p}}\right),\\ 
c_2 & = \gamma_2 \min\left(\sqrt{k_{R}\underline{\lambda}(J)}, \frac{4\underline{\lambda}(J)k_{R}k_{\omega}}{k_{\omega}^2+4\underline{\lambda}(J)k_{R}}\right).
\end{align*}
The above relations and the bounds on $\gamma_{1}$ and $\gamma_{2}$ ensure that conditions \eqref{eq:Boundc1} and \eqref{eq:Boundc2} hold.  Let  $\underline{k}, \overline{k}\in \mathbb{R}_{+}\setminus \{0\}$ be user-defined positive lower and upper bounds on the control gains, respectively, and   $w_1$, $w_2$, and $w_3$ be  positive weights to be assigned to the uniform bounds during the optimization process.
The control gains, and the parameters $\gamma_{1}$ and $\gamma_{2}$ are then determined by solving the following nonlinear optimization problem, where the uniform bounds $\mathcal{L}_{p}(\overline{\mathcal{V}}_{1},\overline{\mathcal{V}}_{2})$, $\mathcal{L}_{v}(\overline{\mathcal{V}}_{1},\overline{\mathcal{V}}_{2})$, and $\mathcal{L}_{f}(\overline{\mathcal{V}}_{1},\overline{\mathcal{V}}_{2})$ are functions of the gains through the relations \eqref{eq:L_u}, \eqref{eq:L_p}, \eqref{eq:L_v}, \eqref{eq:L_f}, and \eqref{eq:V2Bound} and the arguments $\overline{\mathcal{V}}_{1}$ and $\overline{\mathcal{V}}_{2}$ are dropped: 
\begin{equation} \label{eqn:opt_tuning}
\begin{aligned}
\min_{k_p, k_v, k_R, k_\omega, \gamma_1, \gamma_2} \quad & w_1 \mathcal{L}_{p} + w_2 \mathcal{L}_{v} + w_3 \mathcal{L}_{f},\\
\textrm{s.t.} \quad & \underline{k} \leq k_p, k_v, k_R, k_\omega \leq \overline{k}, \\
                    & \quad 0 < \gamma_1, \gamma_2 < 1.
\end{aligned}
\end{equation}

We set the gain bounds to be $\underline{k}=0.1, ~\overline{k}=30$ and choose the weight values $w_1 = 15$, $w_2 = 1$, $w_3 = 1$. 
The gains are obtained by solving \eqref{eqn:opt_tuning} with Simulated Annealing \cite{kirkpatrick1983optimization}, a stochastic global search optimization algorithm. The MATLAB built-in function \texttt{simulannealbnd} is adopted. With the initial guess $$k_p = 10, ~ k_v = 10, ~k_R = 10, ~k_\omega = 10, ~\gamma_1 = 0.5, ~\gamma_2 = 0.5,$$ the optimized control gains and coefficients are obtained as
$$ k_p = 18.5058, ~ k_v = 5.6704, ~k_R = 23.5537, ~k_\omega = 1.4309,$$ $$\gamma_1 = 0.5500, ~\gamma_2 =0.6047.$$
The average computing time for obtaining the gains via optimization is 0.8934 seconds. The resulting value of  $\overline{\mathcal{V}}_{2}$ is  $24.4053$ and the resulting uniform bounds are
$\mathcal{L}_p(\overline{\mathcal{V}}_{1},\overline{\mathcal{V}}_{2})= 0.3374~\text{m}$,
$\mathcal{L}_v(\overline{\mathcal{V}}_{1},\overline{\mathcal{V}}_{2})=  0.6968~ \text{m/s}$, and 
$\mathcal{L}_f (\overline{\mathcal{V}}_{1},\overline{\mathcal{V}}_{2})= 6.2445~ \text{N}$, with 
$\mathcal{L}_u (\overline{\mathcal{V}}_{1},\overline{\mathcal{V}}_{2})= 0.9737$.
In addition, the theoretical bound on thrust is
$$
\overline{\mathcal{F}} \defas \mathcal{L}_f(\overline{\mathcal{V}}_{1},\overline{\mathcal{V}}_{2}) + m\|a_{\max}\| = 50.0763~\text{N},
$$ 
which is less than the specified thrust bound $f_{\max}$.

\subsection{Safe tube and trajectory synthesis}

Next, we construct a safe tube that consists of connected hyper-rectangles, through which the desired trajectory is synthesized. The safe tube is obtained by first constructing  an RRT according to Algorithm \ref{Alg:RRTAlgorithm}, then obtaining a set of waypoints by conducting a shortest path algorithm over the RRT, and finally computing the radii of the hyper-rectangles of the safe tube according to \eqref{eq:RadiiForSafeTube}. The parameters in Algorithm \ref{Alg:RRTAlgorithm} are chosen to be 
$
\alpha = 0.9$,  $N_{v} = 400$, $C_{\mathrm{sampling}} = 0.9$. The resulting safe tube consists of  fifteen hyper-rectangles depicted in Figure \ref{fig:hyperrectangle}, where the associated computational time is  0.06 seconds. 
 
 Next, we construct the desired trajectory given as a piecewise B\'ezier curve with fifteen segments, where each segment is parameterized by fifteen control points ($N_{p}=14$). 
 We implement Algorithm \ref{Alg:DesiredTrajectory} for the trajectory synthesis, where we set   
$\alpha_{T} = 1.1,~T_{0} = 10$, and we use  $\varepsilon=10^{-6}$ for the constraint given in equation \eqref{eq:Constraint10}. The resulting trajectory, with time duration $T = 45.94$ seconds, is depicted in Figure \ref{fig:trajectory_generation}. The trajectory computation based on Algorithm \ref{Alg:DesiredTrajectory} required 0.89 seconds of CPU time.

\begin{figure}[ht]
    \centering
    \includegraphics[width=0.99\columnwidth]{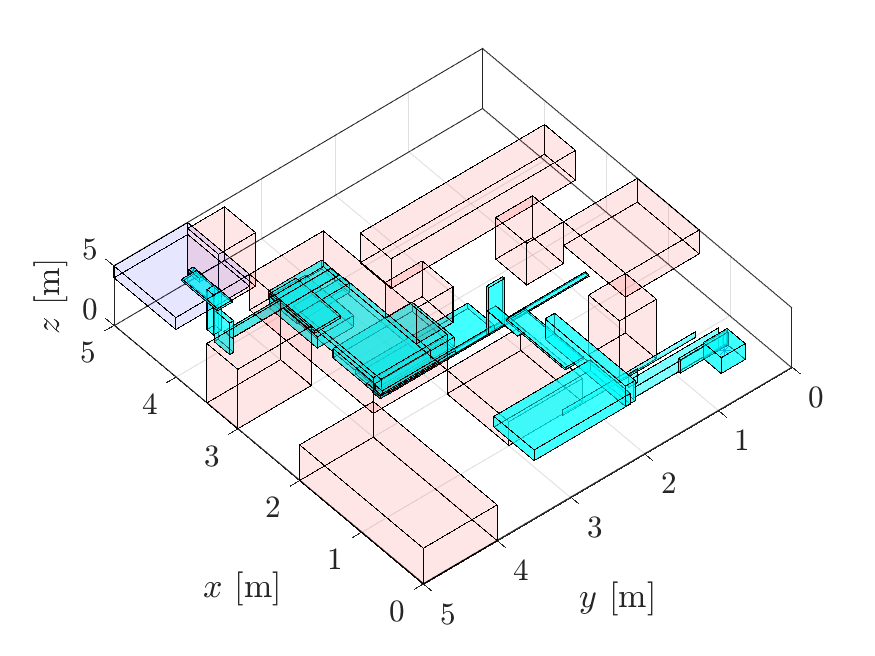}
    \caption{The safe tube is the union of the  cyan boxes. Note that the safe tube does not intersect with the unsafe set (union of the light red boxes), where the last box of the safe tube lies within the target set (blue box).}
    \label{fig:hyperrectangle}
\end{figure}

\begin{figure}[ht]
    \centering
    \includegraphics[width=0.99\columnwidth]{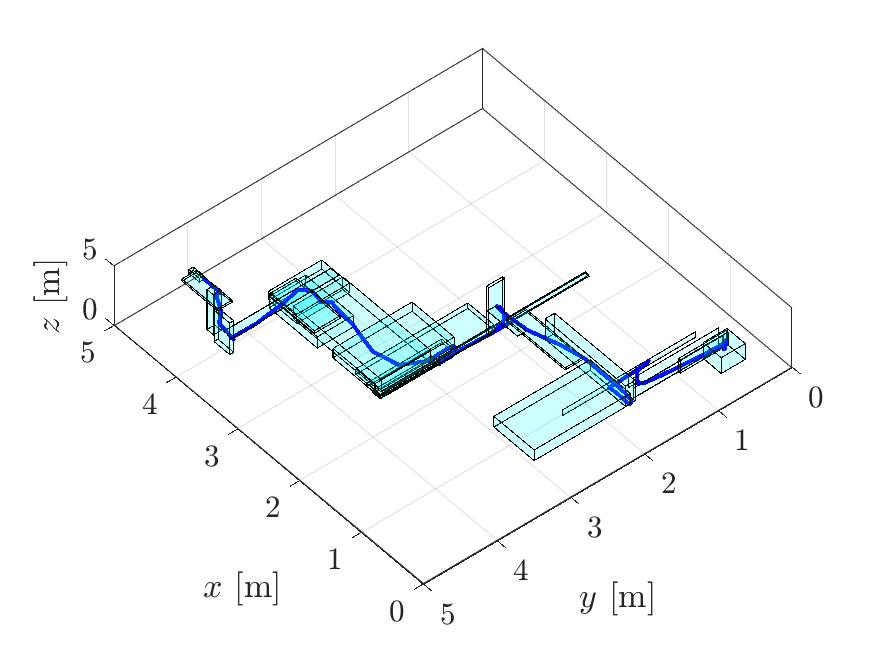}
    \caption{The generated desired trajectory $p_{d}$ is shown as the blue curve, which lies within the safe tube (union of cyan boxes).}
    \label{fig:trajectory_generation}
\end{figure}

\subsection{Initial points generation}

As demonstrated in Theorem \ref{thm:ConditionsForSolvingReachAvoidProblem}, the safety of the closed-loop quadrotor system is guaranteed if the initial conditions satisfy \eqref{eq:InitialSet}. The safe initial set is characterized by \eqref{eq:InitialSet} by means of sampling, where sample points satisfying \eqref{eq:InitialSet} are considered safe, whereas those violating \eqref{eq:InitialSet} are deemed unsafe.  We consider two distinct cases for sampling. In the first case, the initial position error is sampled through the relation $e_p(0) = \mathtt{sample}(0.21\Hintcc{-1_3, 1_3})$, with the initial velocity error $e_v(0)$ attitude error $e_R(0)$, and angular velocity error $e_\omega(0)$ all set to zero. In the second case, the initial attitude error $e_R(0)$ is determined by \eqref{eqn:eW} with  $R(0) = \e{(\mathtt{sample}(0.1\Hintcc{-1_3, 1_3}))^\wedge} \text{(refer to Lemma \ref{lem:RodriguesFormula})}$, while $e_p(0) = e_v(0) = e_\omega(0) = 0$ (This results in $R_{d}(0)=I_{3}$). In each case, one million points are sampled, with safe points marked red and unsafe points marked blue, as illustrated in Figure \ref{fig:initial_shape_ep0_eR0}. The cross-sectional figures shows the boundary between safe region and unsafe region, as represented by the sample points.
\begin{figure}[ht]
    \centering
    \includegraphics[width=0.49\columnwidth]{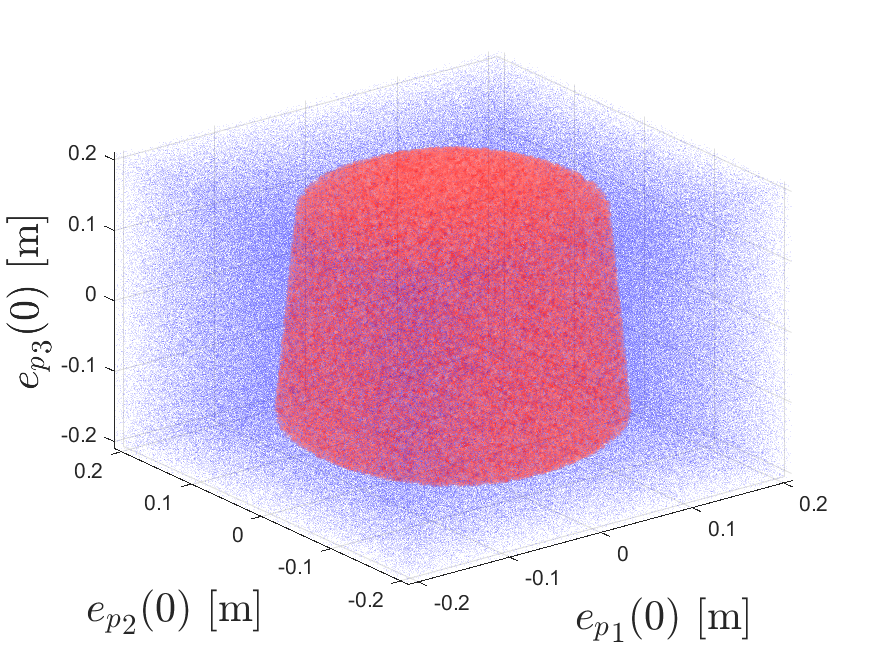}
    \includegraphics[width=0.49\columnwidth]{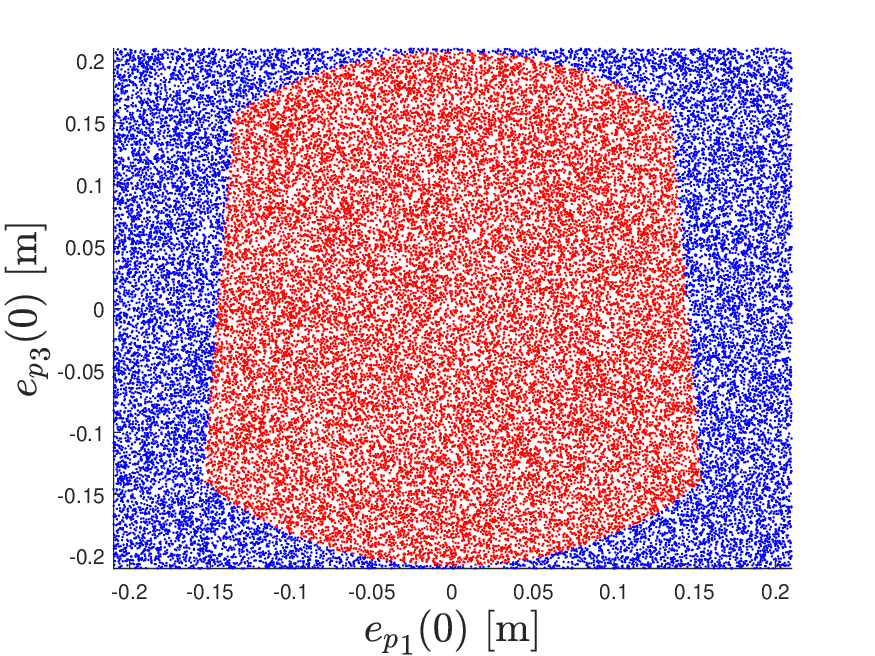}
    \includegraphics[width=0.49\columnwidth]{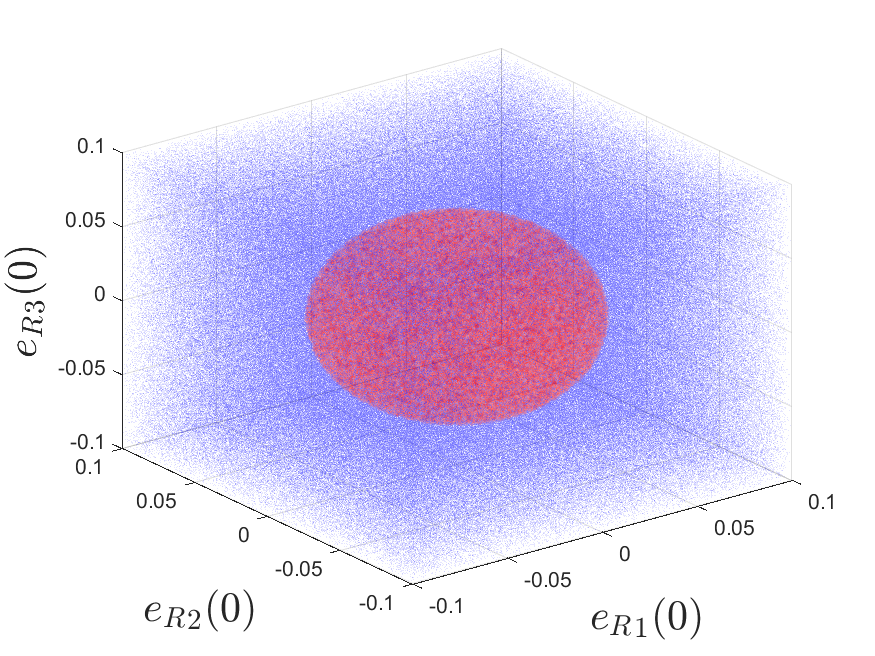}
    \includegraphics[width=0.49\columnwidth]{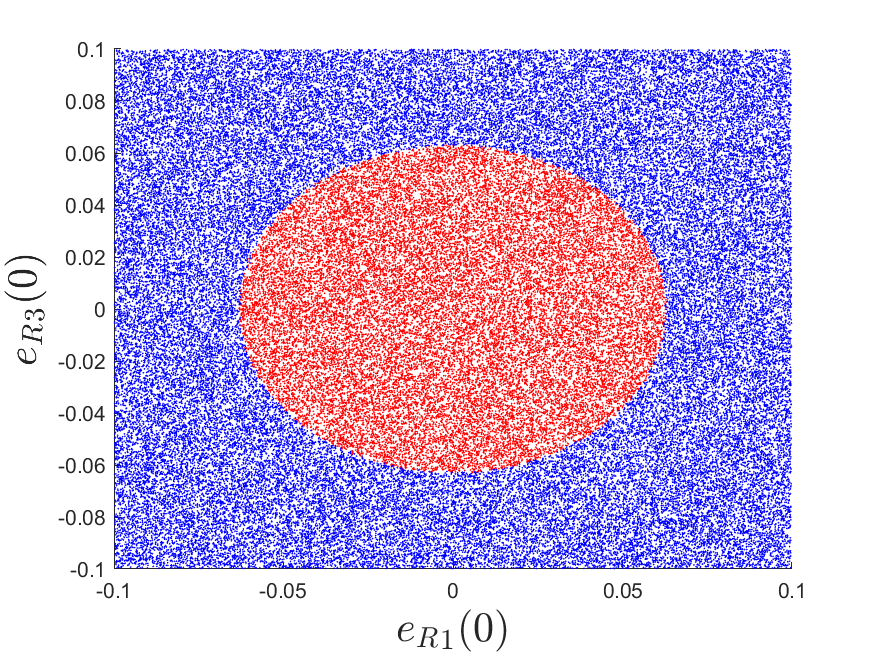}
    \caption{Sample points approximating the safe initial set characterized by \eqref{eq:InitialSet} (the red points  satisfy \eqref{eq:InitialSet}, whereas the blue points do not satisfy \eqref{eq:InitialSet}). The top left panel shows a 3D scatter plot of the sample points with position error assuming $e_v(0) = e_R(0) = e_\omega(0) = 0_{3}$, while the top right panel displays a cross-sectional view when ${e_p}_2(0) = 0$. The bottom left panel shows a 3D scatter plot of the sample points with attitude error assuming $e_p(0) = e_v(0) = e_\omega(0) = 0_{3}$, while the bottom right panel displays a cross-sectional view when ${e_R}_2(0) = 0$.}
    \label{fig:initial_shape_ep0_eR0}
\end{figure}
\subsection{Validation of the tracking performance}
For the sake of numerical quadrotor simulations, we sample twenty initial points through the following relations:
$p(0)= p_d(0) + \mathtt{sample}(0.3\Hintcc{-1_3, 1_3}),$  
$v(0)  =\dot{p}_{d}(0)+ \mathtt{sample}(0.3\Hintcc{-1_3, 1_3}),$ 
$R(0) = \e{(\mathtt{sample}(0.5\Hintcc{-1_3, 1_3}))^\wedge},$ and  
$\omega(0) = \mathtt{sample}(\Hintcc{-1_3, 1_3}),$ where only the  points satisfying \eqref{eq:InitialSet} are considered. The generated twenty initial points are adopted to compute trajectories using 4th and 5th order Runge-Kutta methods (\texttt{RK45} in MATLAB)\footnote{The standard Runge-Kutta method do not generally preserve  the $\SO$ structure of the attitude $R$ during numerical integration. However, the Runge-Kutta methods provide convergence guarantees which motivates using   them in this section. For structure preserving numerical integration methods, see, e.g., \cite{hairer2006geometric}. }. The integrated trajectories are depicted in Figure \ref{fig:traj_simulation}.  All position profiles associated with the generated trajectories from the chosen twenty points, represented by blue lines, originate from the bottom right corner and terminate at the target set (blue box) located at the top left corner.
\begin{figure}[ht]
    \centering
    \includegraphics[width=0.99\linewidth]{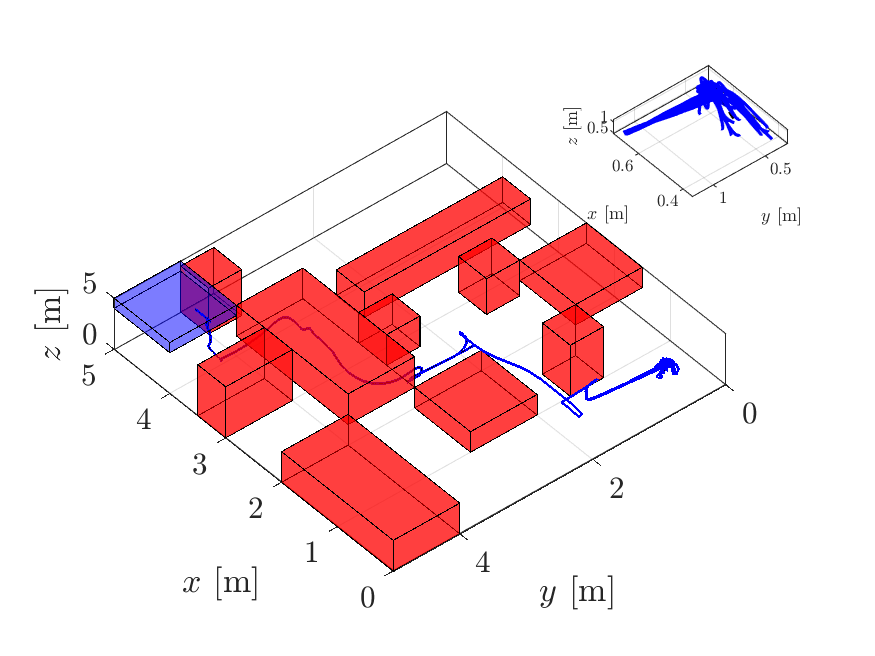}
    \caption{The generated twenty position trajectories going through the operating domain, avoiding the unsafe set (red boxes) and reaching the target set (blue box).  The top right subplot is a zoomed-in view of the trajectories at the initial portion of the simulations.}
    \label{fig:traj_simulation}
\end{figure}

To demonstrate the effectiveness of the proposed framework and the safety guarantees for the generated twenty trajectories, the profiles of the  functions  $V$, $\norm{e_p(\cdot)}$, $\norm{e_v(\cdot)}$, $v$, $\|f(\cdot)\|$, and $F_{d,3}(\cdot)$, are recorded and presented in Figures \ref{fig:pvV}--\ref{fig:fFd3}. The first subplot of Figure \ref{fig:pvV} illustrates that the profiles of the function $V$, associated with the generated trajectories, are bounded below by  $\mathcal{L}_u^2(\overline{\mathcal{V}}_{1},\overline{\mathcal{V}}_{2})$ for all $t \in [0, T]$, thus validating Corollary \ref{cor:V_bound}. The subsequent two subplots of Figure \ref{fig:pvV}  indicate that $\|e_p(\cdot)\|$ and $\|e_v(\cdot)\|$ are bounded by the theoretical bounds $\mathcal{L}_{p}(\overline{\mathcal{V}}_{1},\overline{\mathcal{V}}_{2})$ and $\mathcal{L}_{v}(\overline{\mathcal{V}}_{1},\overline{\mathcal{V}}_{2})$, respectively, for all $t \in [0, T]$, validating Corollary \ref{lem:LpLvLf} and illustrating the collision-avoidance and the velocity bound satisfaction for the generated trajectories.  In fact, Figure \ref{fig:v_max} shows clearly how the absolute values of velocity are within $\Hintcc{0_{3}, v_{\max}}$, ensuring that the tracking remains within the operational capacity of the quadrotor. The absolute values of the thrust $f$, shown in the first subplot of Figure \ref{fig:fFd3}, are uniformly bounded by $\overline{\mathcal{F}}$, indicating  the effectiveness of the thrust bound $\overline{\mathcal{F}}$ and the fulfillment of the requirement of not exceeding $f_{\max}$. The third component of the desired thrust, $F_{d,3}$, is strictly positive, indicating well-definedness of the closed-loop dynamics for all the generated trajectories (see \eqref{eq:Well-DefinednessRequirement}).  The simulations demonstrate that when \eqref{eq:InitialSet} is fulfilled, the quadrotor performs as expected and the position errors remain perfectly within the theoretical threshold, validating the safety guarantees of the proposed framework.
 
\begin{figure}[ht]
    \centering
    \includegraphics[width=0.99\columnwidth]{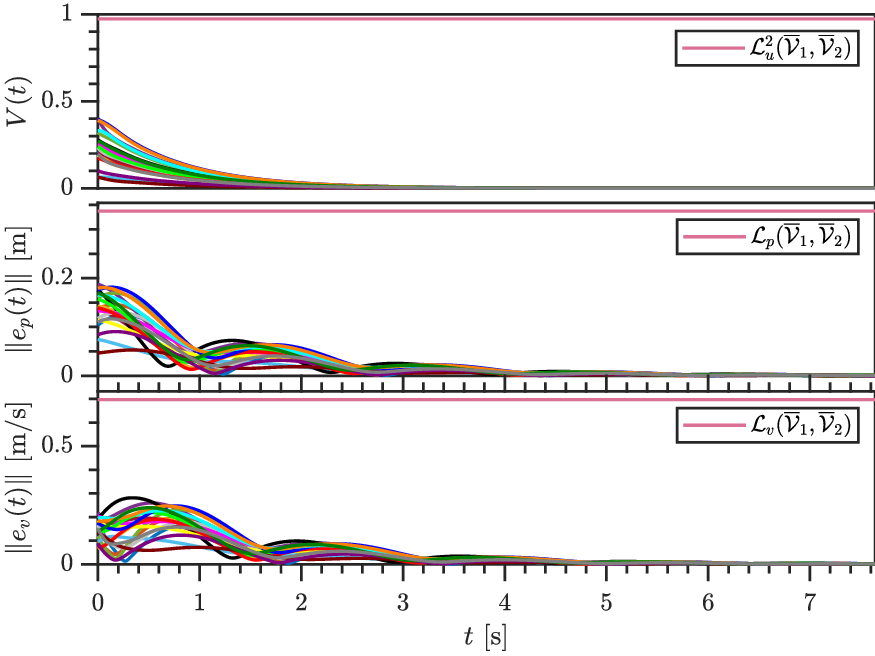}
    \caption{For all $t \in [0, T]$, $\|e_p\|$, $\|e_v\|$, $V$ remain within the theoretical bounds. Data are only shown for the first 7.7 seconds, however the bounds are still respected for all $t\in \intcc{0,T}$.}
    \label{fig:pvV}
\end{figure}

\begin{figure}[ht]
    \centering
    \includegraphics[width=0.99\columnwidth]{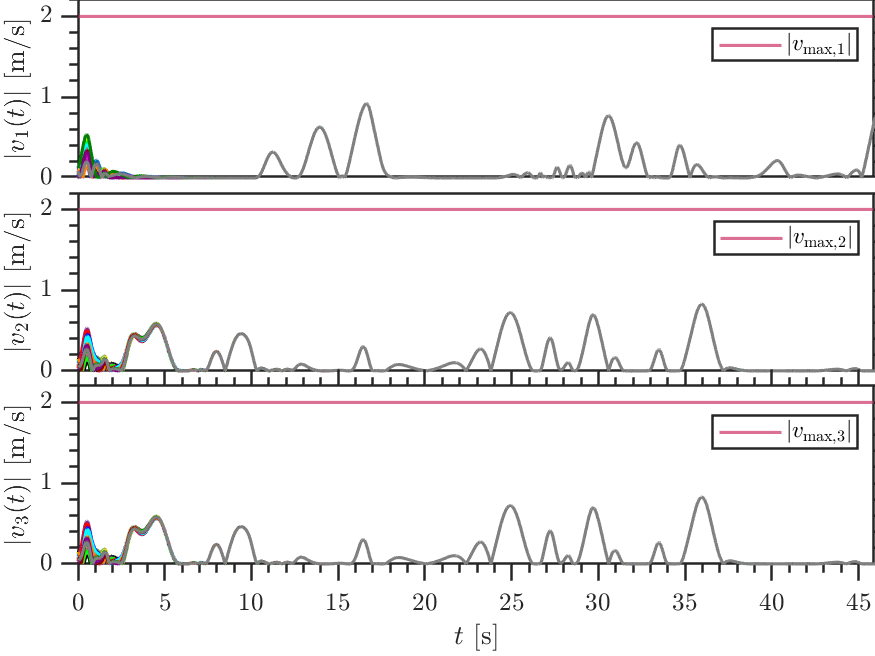}
    \caption{The three components of $\abs{v}$ all remain within the theoretical bound for all $t \in [0, T]$.}
    \label{fig:v_max}
\end{figure}

\begin{figure}[ht]
    \centering
    \includegraphics[width=0.99\columnwidth]{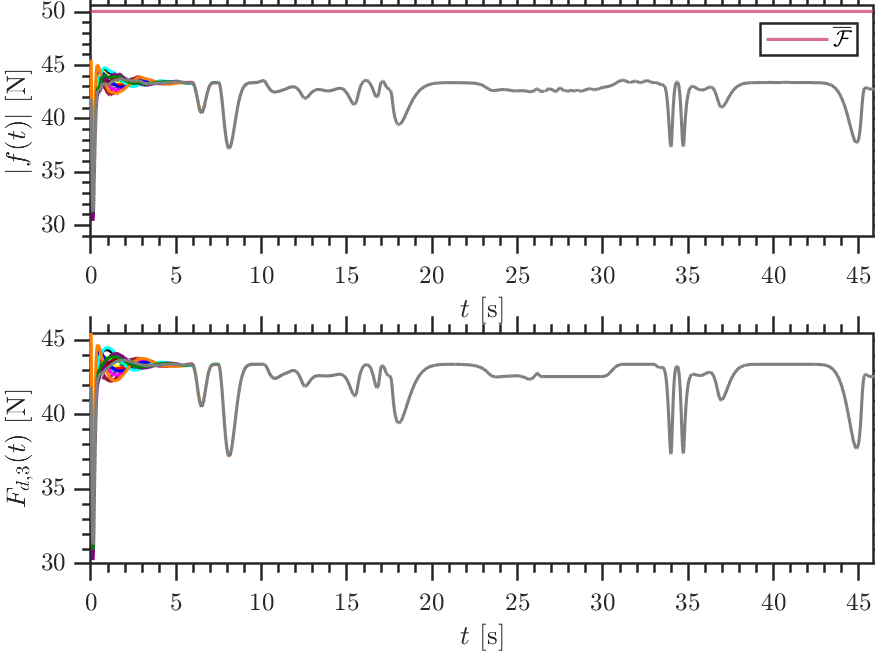}
    \caption{The profiles of $|f|$ stay within the theoretical bound for all the generated trajectories, where the profiles of $F_{d,3}$ are strictly positive, ensuring well-definedness of the closed-loop dynamics.}
    \label{fig:fFd3}
\end{figure}

\section{Conclusion} \label{sec:Conclusion}

In this  paper, we proposed a control framework for a quadrotor UAV to accomplish reach-avoid tasks with formal safety guarantees, where the standard planning-tracking paradigm is adapted to account for tracking errors. The framework integrates geometric control theory for trajectory tracking and polynomial trajectory generation using Bézier curves, where tracking errors are accounted for  during trajectory synthesis.  We revisited the stability analysis of the closed-loop quadrotor system under geometric control, where we proved local exponential stability of the tracking error dynamics for  any positive control gains and we  provided uniform bounds on tracking errors that can be used in planning. We also derived sufficient conditions to be imposed on the desired trajectory to ensure the well-definedness of the closed-loop quadrotor dynamics. The trajectory synthesis involved an efficient algorithm that constructs a safe tube using sampling-based planning and safe hyper-rectangular set computations. The desired trajectory, represented as a piecewise continuous Bézier curve,  is computed through the generated safe tube using a heuristic efficient approach, relying on iterative linear programming. Finally, we performed extensive numerical quadrotor simulations to demonstrate the proposed framework's effectiveness in reach-avoid planning scenarios. 

In future work,  we plan to incorporate the effects of measurement and input noises and disturbances  into  the safe planning  framework to enable effective real-world applications. Additionally, we aim to extend the proposed control synthesis to cover more complex quadrotor models (e.g., those that account for aerodynamic drag) and more general specifications (e.g., those described by temporal logics). There is also a potential to improve the uniform tracking error bounds in future work by employing modified versions of the geometric tracking control that involve gain matrices instead of scalars, resulting in less conservative trajectory synthesis.


\section*{Appendix}
    \subsection*{Proof of Lemma \ref{Lem:EstimatesPSD}}
   For (a), as $M$ is symmetric, it is diagonalizable, i.e.,  $M=Q^{\intercal}DQ$, where $Q\in \mathbb{R}^{n\times n}$ is orthogonal, and $D\in \mathbb{R}^{n\times n}$ is a diagonal matrix whose diagonal entries are the positive eigenvalues of  $M$, $\lambda_{1},\ldots, \lambda_{n}$, i.e., $D=\mathrm{diag}([\lambda_{1}\cdots \lambda_{n}]^{\intercal})$. Let $y=Qx$. Then, 
$$
\underline{\lambda}(M) \norm{y}^{2}\leq x^{\intercal}Mx=y^{\intercal}Dy=\sum_{i=1}^{n}\lambda_{i}y_{i}^{2}\leq \overline{\lambda}(M) \norm{y}^{2}.
$$
As $Q$ is orthogonal, $\norm{y}=\norm{x}$, and (a) follows.
For (b), first  note that as $M\in \mathcal{S}^{n}_{++}$, its square root $M^{\frac{1}{2}}$ is also in $\mathcal{S}^{n}_{++}$. Moreover, $M^{-\frac{1}{2}}\in \mathcal{S}^{n}_{++}$ and, consequently,  $M^{-\frac{1}{2}}WM^{-\frac{1}{2}}\in \mathcal{S}^{n}_{++}$ (see \cite[p.~222]{Abadir2005matrix}).  Therefore, utilizing (a), we have   $
\norm{W^{\frac{1}{2}}x}^{2}=x^{\intercal}W x=x^{\intercal}M^{\frac{1}{2}} M^{-\frac{1}{2}}WM^{-\frac{1}{2}}M^{\frac{1}{2}}x=(M^{\frac{1}{2}}x)^{\intercal} (M^{-\frac{1}{2}}W M^{-\frac{1}{2}}) (M^{\frac{1}{2}}x)\geq  \underline{\lambda}(M^{-\frac{1}{2}}W M^{-\frac{1}{2}}) \norm{M^{\frac{1}{2}}x}^{2}$. The estimate in (c)  follows as
    $
\norm{Ax}=\norm{AM^{-\frac{1}{2}}M^{\frac{1}{2}}x}\leq \norm{A M^{-\frac{1}{2}}}\norm{M^{\frac{1}{2}}x}.
$

\subsection*{Proof of Lemma \ref{lem:SpecificBernoulliInequality}}
 This follows from Lemma \ref{lem:BernoulliInequality}, with $b(\cdot)=-(a_{0}-a_{1}\e{-c(\cdot)})$, and $k(\cdot)=a_{2}\e{-c(\cdot)}$. We have 
 $$
\int_{0}^{t}b(s)\mathrm{d}s= -a_{0}t+a_{1}\frac{1-\e{-ct}}{c}
\leq -a_{0}t+\frac{a_{1}}{c},~t\in I,
 $$
which implies $
\e{\frac{1}{2}\int_{0}^{t}b(s)\mathrm{d}s}\leq \e{-\frac{a_{0}}{2}t}\e{\frac{a_{1}}{2c}}.
$
Also, we have 
$$-\int_{0}^{t}b(s)\mathrm{d}s=a_{0}t+a_{1}\frac{\e{-ct}-1}{c}\leq a_{0}t,~t\in I,
$$
which yields 
\begin{align*}
\int_{0}^{t}k(s)\e{-\frac{1}{2}\int_{0}^{s}b(z)\mathrm{d}z}\mathrm{d}s\leq&  \int_{0}^{t}a_{2}\e{-cs}\e{\frac{a_{0}}{2}s}\mathrm{d}s,~t\in I, 
\end{align*}
and that results in the desired bound.

 \subsection*{Proof of Proposition \ref{Prop:ErrorDynamics}} The proof is adapted from \cite{lee2012robust}. Let $t\in I$. 
 Obviously, $\dot{e}_p(t) = e_v(t)$. The derivative of $e_v$ is given by
 \begin{align*}
     m \dot{e}_v(t) &= m \ddot{p}(t) - m \ddot{p}_d(t) 
     = - mge_3 + f(t)R(t)e_3 - m \ddot{p}_d(t)\\
     &=  - mge_3 - m \ddot{p}_d(t) + F_{d}(t) +(f(t) R(t) e_{3}-F_{d}(t)).
 \end{align*}
 Using the definition of $F_{d}$ in \eqref{eq:Fd}, we have 
 $
 - mge_3 - m \ddot{p}_d(t) + F_{d}(t)=-k_p e_{p}(t)-k_{v}e_{v}(t),
 $
 and $f(t) R(t) e_{3} -F_{d}(t)= (F_{d}(t)\cdot R(t) e_{3})R(t)e_{3} -F_{d}(t)
 =(\norm{F_{d}(t)}b_{3,d}(t)\cdot b_{3}(t))b_{3}(t)-\norm{F_{d}(t)}b_{3,d}(t)
 =\norm{F_{d}(t)}((b_{3,d}(t)\cdot b_{3}(t))b_{3}(t)-b_{3,d}(t))
 =\Delta_{f}(t).$

Now, we derive the evolution equations for the rotation error dynamics. Let $G\colon I\rightarrow \SO$ be defined as $G(\cdot)\defas R_{d}^{\intercal}(\cdot)R(\cdot)$. Note that the rotational error functions can be written in terms of $G$ as follows:
\begin{align*}
e_{R}(\cdot)&=\frac{1}{2}(G(\cdot)+G^\intercal(\cdot))^{\vee},\\
e_{\omega}(\cdot)&=\omega(\cdot)-G^\intercal(\cdot)\omega_{d}(\cdot),\\
\Psi(\cdot)&=\frac{1}{2}\mathrm{tr}(\mathrm{I}_{3}-G(\cdot)).
\end{align*}
In addition, $\tau$ and $\mathcal{C}$
can be written as 
\begin{align*}
   \tau(\cdot)=&- k_R {e_R}(\cdot) -k_{{\omega}} {e_\omega}(\cdot)+{\omega}(\cdot) \times J {\omega}(\cdot)\\ \nonumber
 &  -J(\hat{\omega}(\cdot) G^{\intercal}(t) {\omega_d}(\cdot) - G^{\intercal}(\cdot){\dot{\omega}_d(\cdot)}),
     \\
     \mathcal{C}(\cdot)=&\frac{1}{2}(\tr[G^{\intercal}(\cdot)]\mathrm{I}_{3} - G^{\intercal}(\cdot)).
\end{align*}
Using equations \eqref{eq:Rdot} and \eqref{eq:wd}, and property \eqref{hat0}, the derivative of $G$ is given by   \begin{equation}
\nonumber
     \begin{split}
       \dot{G}(t) & =\dot{R}_{d}^{\intercal}(t)R(t)+R_{d}^{\intercal}(t)\dot{R}(t)\\ &= (R_{d}(t)\hat{\omega}_{d}(t))^{\intercal}R(t)+R_d^{\intercal}(t) R(t) \hat{\omega}(t)
       \\ &=- \hat{\omega}_d(t) R_d^{\intercal}(t) R(t) + R_d^{\intercal}(t) R(t) \hat{\omega}(t)\\
         & =  R_d^{\intercal}(t) R(t)\left( \hat{\omega}(t) -  R^\intercal(t) R_d(t) \hat{\omega}_d(t) R_d^{\intercal}(t) R(t)\right)\\
         &= G(t)\left( \hat{\omega}(t) -  G^{\intercal}(t) \hat{\omega}_d(t) G(t)\right).
     \end{split}
     \end{equation}
     Using  property \eqref{hat4}, we have 
     $G^{\intercal}(t) \hat{\omega}_d(t) G(t)=(G^{\intercal}(t) \omega_d(t))^\wedge$. Hence, using the linearity of the hat operator and the definition of $e_\omega$ in \eqref{eqn:eW}, we have
     \begin{equation}\nonumber
    \begin{split}
         \dot{G}(t) & = G(t) (\hat{\omega}(t) - (G^{\intercal}(t)\omega_d(t))^\wedge)\\
         & = G(t) (\omega(t) - G^{\intercal}(t) \omega_d(t))^\wedge=G(t)\hat{e}_\omega(t).
     \end{split}
     \end{equation}

     For the time derivative of $\Psi$, we have
     \begin{equation}
        \nonumber \begin{split}
            \dot{\Psi}(t) & = \frac{1}{2}\tr[0_{3\times 3} -\dot{G}(t)] = -\frac{1}{2}\tr[G(t)\hat{e}_\omega(t)].
         \end{split}
     \end{equation}
     Using property \eqref{hat2} and the definition of $e_{R}$, we have
     \begin{equation}
        \nonumber \begin{split}
             \dot{\Psi}(t) & = \frac{1}{2}e_\omega^\intercal(t)(G(t) - {G}^{\intercal}(t))^\wedge = e_\omega^\intercal(t) e_R(t)= e_\omega(t) \cdot e_R(t).         \end{split}
     \end{equation}

     For the time derivative of $e_R$, we have, where we use properties \eqref{hat0} and \eqref{hat3},
 \begin{equation}
    \nonumber
     \begin{split}
         \dot{e}_R(t) & = \frac{1}{2} (\dot{G}(t) - \dot{G}^{\intercal}(t))^\vee= \frac{1}{2} (G(t) \hat{e}_\omega (t)- \hat{e}_\omega^{\intercal}(t)G^{\intercal}(t))^\vee\\
         & = \frac{1}{2} (G(t) \hat{e}_\omega(t) + \hat{e}_\omega(t)G^{\intercal}(t))^\vee\\
         & = \frac{1}{2}(\tr[G^{\intercal}(t)]\mathrm{I}_{3} - G^{\intercal}(t)) e_\omega(t)= \mathcal{C}(t)e_{\omega}(t).
     \end{split}
 \end{equation}
 For the time derivative of $e_{\omega}$, noting that $e_{\omega}(t)=\omega(t)-G^\intercal(t) \omega_{d}(t)$, we have
 \begin{equation}
  \nonumber   \begin{split}
        \dot{e}_\omega(t) & = \dot{\omega}(t) - \dot{G}^{\intercal}(t)\omega_{d}(t)-G^{\intercal}(t)\dot{\omega}_{d}(t)\\
         & = \dot{\omega}(t) - \hat{e}_\omega^{\intercal}(t) G^{\intercal}(t)\omega_d(t) -G^{\intercal}(t)\dot{\omega}_d(t)\\
         &= \dot{\omega}(t) + \hat{e}_\omega(t) G^{\intercal}(t) \omega_d(t) -G^{\intercal}(t)\dot{\omega}_d(t)
             .
    \end{split}
 \end{equation}
 Substituting equation \eqref{eq:wdot} in yields
 \begin{align*}
     \dot{e}_{\omega}({t})=& J^{-1} \left(-{\omega(t)} \times J {\omega(t)} +{\tau(t)}\right)\\
     &+\hat{e}_\omega(t) G^{\intercal}(t) \omega_d(t) -G^{\intercal}(t)\dot{\omega}_d(t).
 \end{align*}
 Substituting the expression of $\tau$ in the formula of $\dot{e}_{\omega}$ yields
 \begin{align*}
  \dot{e}_{\omega}(t)=& J^{-1}(-k_{R}e_{R}(t)-k_{\omega}e_{\omega}(t))\\
  &- \hat{\omega}(t) G^{\intercal}(t) {\omega_d}(t)+\hat{e}_\omega(t) G^{\intercal}(t) \omega_d(t). 
 \end{align*}
 Note that, using the linearity of the hat operator and property \eqref{hat4}, we have 
$
\hat{e}_\omega(t)= ({\omega}(t) -G^{\intercal}(t) {\omega}_d(t))^{\wedge}= \hat{\omega}(t)-(G^{\intercal}(t) {\omega}_d(t))^{\wedge}= \hat{\omega}(t)-G^{\intercal}(t)\hat {\omega}_d(t)G(t)$.
Hence,
\begin{align*}
- \hat{\omega}(t) G^{\intercal}(t) {\omega_d}(t)+&\hat{e}_\omega(t) G^{\intercal}(t) \omega_d(t)\\
&=-G^{\intercal}(t)\hat {\omega}_d(t)G(t)G^{\intercal}(t)\omega_{d}(t)\\
&= -G^{\intercal}(t)\hat{\omega}_{d}(t)\omega_{d}(t)=0_{3}
\end{align*}
as $\hat{\omega}_{d}(t)\omega_{d}(t)=\omega_{d}(t)\times \omega_{d}(t)=0_{3
}$, and that completes the proof.

\subsection*{Proof of Lemma \ref{lem:BoundingC}}
This proof is adapted from \cite{lee2010control}. 
    Using Lemma \ref{lem:RodriguesFormula}, let $x\in \mathbb{R}^{3}$ be such that $A=R^{\intercal}(t)R_{d}(t)=\exp(\hat{x})$ (assume without loss of generality that $x\neq 0_{3}$). Then, it can be shown that $\mathrm{tr}(A)=1+2\cos(\norm{x})$ and that $\frac{A+A^{\intercal}}{2}$ has the eigenvalues $1$ and $\cos(\norm{x})$ (repeated). In addition, $\mathcal{C}^{\intercal}(t)\mathcal{C}(t)$ can be written as 
    $$
    \mathcal{C}^{\intercal}(t)\mathcal{C}(t)=\frac{1}{4}\left((\mathrm{tr}^2(A)+1)\mathrm{I}_{3}-2\mathrm{tr}(A)(\frac{A+A^{\intercal}}{2})\right).
    $$ 
   This indicates that the eigenvalues of $\mathcal{C}^{\intercal}(t)\mathcal{C}(t)$ are 
    \begin{align*}
    \lambda_{1}(\mathcal{C}^{\intercal}(t)\mathcal{C}(t))&=\frac{1}{4}\left((1+2\cos(\norm{x}))^{2}+1
 -2(1+2\cos(\norm{x}))\right)\\
 &=\cos^{2}(\norm{x})
   \end{align*}
   and
   \begin{align*}
    \lambda_{2}(\mathcal{C}^{\intercal}(t)\mathcal{C}(t))&=\frac{1}{4}(1+2\cos(\norm{x}))^{2}\\
    &+\frac{1}{4}\left(1-2(1+2\cos(\norm{x}))\cos(\norm{x})\right)\\ &=\frac{1+\cos(\norm{x})}{2} ~\text{(repeated)}, 
 \end{align*}
 which are less than or equal to one, and that completes the proof.
\subsection*{Proof of Lemma \ref{lem:BoundingPsiUsingEr}}
  The proof is adapted from \cite{lee2010control}. Using Lemma \ref{lem:RodriguesFormula}, let $x\in \mathbb{R}^{3}$ such that $R_{d}^{\intercal}(t)R(t)=\exp(\hat{x})$ (without loss of generality, assume $x\neq 0_{3}$). Then, we have 
  $\Psi(t)=1-\cos(\norm{x}),
  $
and 
  $
  e_{R}(t)={\sin(\norm{x})}x/{\norm{x}},
  $
implying $\norm{e_{R}(t)}^{2}=\sin^{2}(\norm{x})$. It then follows that 
 $
 0\leq \Psi(t)\leq 2$, $\norm{e_{R}(t)}\leq 1$, 
 and $
 \norm{e_{R}(t)}^{2}=\sin^{2}(\norm{x})=1-\cos^{2}(\norm{x})
 = (1-\cos(\norm{x}))(1+\cos(\norm{x}))
 = (1-\cos(\norm{x}))(2-(1-\cos(\norm{x})))
 =\Psi(t)(2-\Psi(t))$. 
  Assume without loss of generality that $\norm{x}\neq (2n+1)\pi,~n\in \mathbb{Z}_{+}$, then 
 $
{\norm{e_{R}(t)}^{2}}/{2}\leq {\norm{e_{R}(t)}^{2}}/(2-\Psi(t))=\Psi(t).
 $
 With the assumption that $\Psi(t)\leq \psi<2$, we have 
 $
\Psi(t)={\norm{e_{R}(t)}^{2}}/(2-\Psi(t))\leq {\norm{e_{R}(t)}^{2}}/(2-\psi).
 $

\subsection*{Proof of Lemma \ref{lem:BoundingTheDifferenceBetweenB3andB3d}}
Let $A=R_{d}^{\intercal}(t)R(t)$. We have   $
\norm{(b_{3,d}(t)\cdot b_{3}(t))b_{3}(t)-b_{3,d}(t)}^{2}=(b_{3,d}(t)\cdot b_{3}(t))^{2}b_{3}(t)\cdot b_{3}(t)-2(b_{3,d}(t)\cdot b_{3}(t))^2+b_{3,d}(t)\cdot b_{3,d}(t)\\
=
1-((b_{3,d}(t)\cdot b_{3}(t)))^{2}=1-(A_{3,3})^{2}$.
Using Lemma \ref{lem:RodriguesFormula}, let $x\in \mathbb{R}^{3}$ such that $A=\exp(\hat{x})$ (without loss of generality, assume $x\neq 0_{3}$). As we showed in the proof of Lemma \ref{lem:BoundingPsiUsingEr}, we have $\norm{e_{R}(t)}^{2}=\sin^2(\norm{x})$ and $\Psi(t)=1-\cos(\norm{x})$. The expression of $1-(A_{3,3})^{2}$ is given as  
$$
  1-A_{3,3}^{2}=1-\left(1+\frac{(x_{1}^{2}+x_{2}^{2})(\cos(\norm{x})-1)}{\norm{x}^{2}}\right)^{2}. $$
Let 
$
\alpha_{x}\defas (x_{1}^{2}+x_{2}^{2})/{\norm{x}^{2}}.
$
Note that $0\leq \alpha_{x}\leq 1$. The term $1-(A_{3,3})^{2}$ can then be rewritten as 
$1-(A_{3,3})^{2}=1-(1-\alpha_{x}\Psi(t))^{2}= \alpha_{x}\Psi(t)(2-\alpha_{x}\Psi(t))$. As $0\leq \Psi(t)\leq \psi<2$ and $0\leq \alpha_{x}\leq 1$, we have, using Lemma \ref{lem:BoundingPsiUsingEr}, 
 $
\alpha_{x}\Psi(t)(2-\alpha_{x}\Psi(t))\leq 2\alpha_{x}\Psi(t)\leq 2\Psi(t)\leq \frac{2}{2-\psi}\norm{e_{R}(t)}^2.
 $


\subsection*{Proof of Lemma \ref{lem:ConditionsOnPd}}
    We have $e_{p}(0)=e_{v}(0)=0_{3}$. Consequently, and using the definition of $R_{d}$ in \eqref{eq:Rd} and fact that $\ddot{p}_d(0)=0_{3}$, we have $R_{d}(0)=R(0)=\mathrm{I}_{3}$. Finally, we note that
    \begin{equation}\label{eq:Fd_dot}
\dot{F}_{d}(t)=-k_p e_{v}(t)-\frac{k_v}{m}(-k_{p}e_{p}(t)-k_{v}e_{v}(t)+\Delta_{f}(t))+\dddot{p}_{d}(t), 
\end{equation}
   $t\in I$. As $R(0)=R_{d}(0)$, we have $\Delta_{f}(0)=0$, and with the assumption on $\dddot{p}_{d}(0)$, we have  
     $\dot{F}_{d}(0)=0_{3}$, implying $\dot{R}_{d}(0)=0_{3\times 3}$. Finally, using the definition of $\omega_{d}$ in \eqref{eqn:eW}, the proof is complete.

\subsection*{Proof of Lemma \ref{lem:SafeBox}}
 The proof is taken from \cite{Serry2024Safe}. For any $z\in \mathbb{R}^{n}$, with $\abs{z-v}\leq r$,
    $
\abs{z-\mathrm{center}(\Hintcc{a,b})}\leq \abs{z-v}+\abs{v-\mathrm{center}(\Hintcc{a,b})}\leq  \mathrm{radius}(\Hintcc{a,b})-\abs{\mathrm{center}(\Hintcc{a,b})-v}+\abs{\mathrm{center}(\Hintcc{a,b})-v}=\mathrm{radius}(\Hintcc{a,b})
    $.

\subsection*{Proof of Lemma \ref{lem:SafeStrip}}
    The proof is taken from \cite{Serry2024Safe}. The first claim follows by noting that, for all $y\in \Hintcc{a,b}$, $\abs{x_{i}-y^{\ast}_{i}}\leq \abs{x_{i}-y_{i}}$, $i\in \intcc{1;n}$. For the second claim, let $z\in \mathcal{S}(x,\Hintcc{a,b},\alpha)$, then $\abs{z_{\tilde{i}}-x_{\tilde{i}}}\leq \alpha \norm{x-y^{\ast}}_{\infty}<\norm{x-y^{\ast}}_{\infty}$.  Assume, without loss of generality, that $x_{\tilde{i}}\geq c_{\tilde{i}}$, then it  holds, using the definition and minimal property of $y^{\ast}$, that $x_{\tilde{i}}-c_{\tilde{i}}\geq \norm{x-y^{\ast}}_{\infty}+r_{\tilde{i}}$. Consequently, $z_{\tilde{i}}-c_{\tilde{i}}=z_{\tilde{i}}-x_{\tilde{i}}+x_{\tilde{i}}-c_{\tilde{i}}\geq -\abs{z_{\tilde{i}}-x_{\tilde{i}}}+\norm{x-y^{\ast}}_{\infty}+r_{\tilde{i}}>-\norm{x-y^{\ast}}_{\infty}+\norm{x-y^{\ast}}_{\infty}+r_{\tilde{i}}=r_{\tilde{i}}$. Hence, $\abs{z_{\tilde{i}}-c_{\tilde{i}}}>r_{\tilde{i}}$, implying $z\not\in \Hintcc{a,b}$.

\bibliographystyle{ieeetr}

\end{document}